\documentclass[12pt]{article}
\usepackage{amsmath, amsthm}
\usepackage{url}
\usepackage[a4paper, total={6.4in, 9in}]{geometry}
\usepackage{graphicx}
\usepackage{color}
\usepackage{mathtools}
\usepackage{makecell}
\usepackage{enumitem}
\usepackage[font=small]{caption}
\graphicspath{{images/}}

\newcommand*{\key}{\emph{key}}
\newcommand*{\cT}[0]{\mathcal{T}}
\newcommand*{\OO}{\mathrm{O}}
\newcommand*{\barM}{\overline{M}}
\newcommand*{\Decreasekey}{\mbox{\it decrease-key}}
\newcommand*{\Delete}{\mbox{\it delete}}
\newcommand*{\Deletemin}{\mbox{\it delete-min}}
\newcommand*{\Findmin}{\mbox{\it find-min}}
\newcommand*{\Insert}{\mbox{\it insert}}
\newcommand*{\Makeheap}{\mbox{\it make-heap}}
\newcommand*{\Meld}{\mbox{\it meld}}

\newtheorem{theorem}{Theorem}
\newtheorem{lemma}[theorem]{Lemma}
\newtheorem{corollary}[theorem]{Corollary}
\newtheorem{remark}[theorem]{Remark}
\numberwithin{theorem}{section}

\begin{document}
\title{Efficiency of Self-Adjusting Heaps}
% \author{Corwin Sinnamon}
% \affiliation{\institution{Princeton University} \department{Department of Computer Science} \city{Princeton} \state{NJ} \country{USA}}
% \email{sinncore@gmail.com}
% \author{Robert E. Tarjan}
% \affiliation{\institution{Princeton University} \department{Department of Computer Science} \city{Princeton} \state{NJ} \country{USA}}
% \email{ret@cs.princeton.edu}
% \date{}

% \acmJournal{TALG}
% \acmSubmissionID{TALG-2023-0073.R1}

\author{Corwin Sinnamon\thanks{Department of Computer Science, Princeton University.  Research partially supported by a gift from Microsoft.} \and Robert E. Tarjan\footnotemark[1]}
\date{}

\maketitle

\begin{abstract}
Since the invention of the pairing heap by Fredman, Sedgewick, Sleator, and Tarjan~\cite{FSST86}, it has been an open question whether this or any other simple ``self-adjusting" heap supports decrease-key operations in $\OO(\log\log n)$ time, where $n$ is the number of heap items.  Using powerful new techniques, we answer this question in the affirmative. We prove that both slim and smooth heaps, recently introduced self-adjusting heaps, support heap operations in the following amortized time bounds: $\OO(\log n)$ for delete-min and delete, $\OO(\log\log n)$ for decrease-key, and $\OO(1)$ for all other heap operations, including insert and meld, where $n$ is the number of heap items that are eventually deleted: Items inserted but never deleted do not count in the bounds.  We also analyze the multipass pairing heap, a variant of pairing heaps.  For this heap implementation, we obtain the same bounds except for decrease-key, for which our bound is $\OO(\log\log n \cdot\log\log\log n)$, where again items that are never deleted do not count in $n$.  Our bounds significantly improve the best previously known bounds for all three data structures.  For slim and smooth heaps our bounds are tight, since they match lower bounds of Iacono and {\"O}zkan~\cite{IaconoOzkan}.
\end{abstract}

% \begin{CCSXML}
% <ccs2012>
%    <concept>
%        <concept_id>10003752.10003809</concept_id>
%        <concept_desc>Theory of computation~Design and analysis of algorithms</concept_desc>
%        <concept_significance>500</concept_significance>
%        </concept>
%    <concept>
%        <concept_id>10003752.10003809.10010031</concept_id>
%        <concept_desc>Theory of computation~Data structures design and analysis</concept_desc>
%        <concept_significance>500</concept_significance>
%        </concept>
%  </ccs2012>
% \end{CCSXML}

% \ccsdesc[500]{Theory of computation~Design and analysis of algorithms}
% \ccsdesc[500]{Theory of computation~Data structures design and analysis}

% \keywords{Heap, Priority Queue, Analysis of Algorithms, Amortized Analysis, Pairing Heap, Multipass Pairing Heap, Slim Heap, Smooth Heap}

% \maketitle
% \begin{acks}
% Research partially supported by a gift from Microsoft.  Part of this work was done while Robert Tarjan was visiting the Simons Institute for the Theory of Computing.
% \end{acks}

\section{Introduction}
\label{S:introduction}

A \emph{heap} (or \emph{priority queue}) is a data structure that stores a set of items, each with a \emph{key} selected from a totally ordered key space. Heaps support the following operations:

\begin{itemize}
\setlength\itemsep{1em}
\item[] $\Makeheap()$: Create and return a new, empty heap.

\item[] $\Findmin(H)$: Return an item of smallest key in heap $H$, or  null if $H$ is empty.

\item[] $\Insert(H, e)$: Insert item $e$ with predefined key into heap $H$.  Item $e$ must not be in any other heap.

\item[] $\Deletemin(H)$: Delete $\Findmin(H)$ from $H$.  Heap $H$ must be nonempty before the delete-min.

\item[] $\Meld(H_1, H_2)$: Return a heap containing all items in item-disjoint heaps $H_1$ and $H_2$, destroying $H_1$ and $H_2$ in the process.

\item[] $\Decreasekey(H, e, k)$: Given a pointer to item $e$ in heap $H$, and given that the key of $e$ is at least $k$, change the key of $e$ to $k$. 

\item[] $\Delete(H, e)$: Given a pointer to item $e$ in heap $H$, delete $e$ from $H$. 
\end{itemize}

One can implement $\Delete(H, e)$ as $\Decreasekey(H, e, -\infty)$ followed by $\Deletemin(H)$, where $-\infty$ denotes a dummy key smaller than all those in the key space.   We shall assume this implementation and not further mention $\Delete$: In all the data structures we consider the time bound of $\Delete$ is the same as that of $\Deletemin$ to within a constant factor.

In discussing efficiency bounds for heap operations, we shall denote by $n$ the number of items in the heap or heaps involved in the operation.  To simplify certain bounds we shall assume that $n \geq 8$.  We denote $\log_2 \max\{n, 1\}$ by $\lg n$.  

There is a vast literature on heap implementations.  See~\cite{BrodalSurvey} for a survey.  If the only operations on keys are comparisons, a sequence of $n$ inserts followed by $n$ delete-mins will sort $n$ numbers.  Thus either insert or delete-min must take $\Omega(\lg n)$ amortized time.  Most of the comparison-based heap implementations in the literature support all heap operations in $\OO(\lg n)$ time per operation, possibly amortized. 

For certain applications, notably Dijkstra's shortest path algorithm, this bound may not be good enough.  On an $n$-vertex, $m$-arc graph, Dijkstra's algorithm does at most $n$ inserts and at most $n$ delete-mins, but the number of decrease-key operations can be as many as $m-n+1$.  Making decrease-key operations faster can thus speed up Dijkstra's algorithm for certain graph densities.

The \emph{Fibonacci heap}~\cite{Fibonacci} was designed to be optimal for such an application.  In a Fibonacci heap, each delete-min takes $\OO(\lg n)$ amortized time, but every other operation, including decrease-key, takes $\OO(1)$ amortized time.  Hence Dijkstra's algorithm runs in $\OO(m+n\lg n)$ time if implemented using a Fibonacci heap.  In fact, the Fibonacci heap sets the theoretical standard: its amortized bounds are best possible for any comparison-based heap. In practice, Fibonacci heaps are not especially efficient, however.  Their theoretical efficiency comes from maintaining a certain kind of balance.  The time and extra space needed to maintain balance makes them slower in practice than simpler implementations.

A natural question, then, is whether there is a heap implementation that has the same (optimal) amortized bounds as Fibonacci heaps and is simpler as well as faster in practice.  The \emph{pairing heap}~\cite{FSST86} was devised as an answer to this question.  The inventors of this data structure proved that pairing heaps take $\OO(1)$ time per make-heap and find-min (both amortized and worst-case) and take $\OO(\lg n)$ amortized time for each other operation.  Iacono~\cite{IaconoPairing} later reduced the bound for insert and meld to $\OO(1)$, a result Sinnamon and Tarjan~\cite{SimplerPairing} recently simplified.  Pettie~\cite{PettiePairing}, in work predating Iacono's, obtained an amortized bound of $\OO(4^{\sqrt{\lg\lg N}})$ for insert, meld, and decrease-key, where $N$ is the maximum number of items ever in a heap, while preserving the $\OO(\lg n)$ bound for delete-min.  Sinnamon~\cite{phd2022} simplified and slightly improved Pettie's result, reducing the amortized bound for insert, meld, and decrease-key to $\OO(\sqrt{\lg\lg N}2^{\sqrt{2\lg\lg N}})$.

The known bounds for decrease-key may be far from tight: The inventors of pairing heaps conjectured that they have the same amortized efficiency as Fibonacci heaps, which if true would make decrease-key an $\OO(1)$-time operation.  But Fredman~\cite{FredmanLB} disproved this conjecture by showing that pairing heaps and similar data structures need $\Omega(\lg\lg n)$ amortized time per decrease-key if insert and delete-min take $\OO(\lg n)$ amortized time.  Whether the efficiency of pairing heaps matches Fredman's lower bound remains an open question.  

Nevertheless, the pairing heap is simple, robust, and efficient in practice. It is \emph{self-adjusting}: It does not explicitly maintain a ``balanced'' state, ready to serve any incoming request quickly at all times, but rather reorganizes itself in a simple, uniform way that guarantees efficiency in an amortized sense, i.e., over a long sequence of operations. In general a self-adjusting data structure stores little or no auxiliary information about its state -- it does not need to, since it does not enforce any specific structure.  Besides saving space, self-adjusting heaps are often efficient in practice and easy to implement, as compared to ``balanced" heaps such as Fibonacci heaps~\cite{LarkinSenTarjan}.

The question remains: Is there \emph{any} self-adjusting heap that is both practical and theoretically competitive with Fibonacci heaps?  Any heap that satisfies the conditions of Fredman's lower bounds cannot match the efficiency of Fibonacci heaps on decrease-key operations.  Moreover, an $\Omega(\lg\lg n)$ amortized lower bound for decrease-key was proved by Iacono and {\"O}zkan~\cite{IaconoOzkan} for the \emph{pure heap} model, which captures another large class of self-adjusting heaps.  The self-adjusting heap models of Fredman and of Iacono and {\"O}zkan do not cover all the possibilities.  In particular, they assume that the decrease-key operation is done by cutting a node and its subtree from the existing tree, as we describe in Section~\ref{S:canonical-framework}, but there are other possibilities, such as removing the item from its node and putting it in a new node, as in hollow heaps~\cite{HollowHeaps}.  Indeed, the term ``self-adjusting heap'' does not have a precise technical definition.   Nevertheless, these lower bounds do provide some evidence that no heap that is intuitively self-adjusting and natural can support decrease-key in $\OO(1)$ time.

Elmasry~\cite{Elmasry09,Elmasry17} presented a modification of pairing heaps that \emph{does} support decrease-key in $\OO(\lg\lg n)$ amortized time.  It keeps any items that have undergone a decrease-key in a buffer, and periodically groups and sorts the buffer items by key to decide how to restructure the heap.  Elmasry's bound for inserts is $\OO(1)$; his bound for melds is $\OO(\lg\lg n)$.  Elmasry's heap is not subject to either of the lower bounds, however, and sorting buffer groups is arguably not ``self-adjusting''.

We ask instead whether there is a self-adjusting heap subject to either Fredman's or Iacono and {\"O}zkan's lower bounds that matches these bounds, i.e., that takes $\OO(\lg n)$ time per delete-min, $\OO(\lg\lg n)$ time per decrease-key, and $\OO(1)$ time for each other heap operation.  Although many heap implementations have been invented and analyzed since the invention of pairing heaps, no work until ours has answered this question.

We analyze three self-adjusting heaps: multipass pairing heaps, slim heaps, and smooth heaps.

Multipass pairing heaps are a variant of pairing heaps proposed and analyzed in the original pairing heaps paper~\cite{FSST86}.  The amortized bounds obtained there are $\OO(1)$ for make-heap and find-min, $\OO(\lg n \cdot \lg\lg n/\lg\lg\lg n)$ for delete-min, and $\OO(\lg n/\lg\lg\lg n)$ for insert, meld, and decrease-key.\footnote{The paper~\cite{FSST86} only claims $\OO(\lg n \lg\lg n/\lg\lg\lg n)$ for each operation, but the potential function used there yields the better bounds we have stated.}  These bounds are not tight: 32 years later, Dorfman, Kaplan, Kozma, Pettie, and Zwick~\cite{DKKPZ} proved that delete-min takes $\OO(\lg n \cdot \log^*n \cdot 2^{\log^*N})$ amortized time, where $N$ is the maximum number of items ever in a heap, and $\log^*$ is the iterated $\log$ function.  They did not explicitly address the other heap operations, since delete-min was their focus, but their potential function yields bounds of $\OO(\lg n/(\lg\lg n)^2 \cdot 2^{\log^*N})$ for insert, meld, and decrease-key, and $\OO(1)$ for find-min and make-heap. 

We obtain the following amortized time bounds for operations on multipass pairing heaps: $\OO(\lg n)$ for delete-min, $\OO(\lg\lg n \cdot\lg\lg\lg n)$ for decrease-key, and $\OO(1)$ for the other heap operations.  This result \emph{almost} answers our question: only the bound for decrease-key differs from the lower bounds, by a factor of $\lg\lg\lg n$.  Whether this factor is inherent or can be eliminated by a better analysis we leave as an open problem.

To completely answer our question, we analyze two more-recently invented self-adjusting heaps, \emph{smooth heaps} and \emph{slim heaps}.  Smooth heaps were invented by Kozma and Saranurak~\cite{KS19} as one result of their study of a duality between heaps and binary search trees (BST's), in which they showed that an execution of a sequence of operations on a binary search tree corresponds to a dual execution of a related sequence of heap operations, with the two sequences taking the same time to within a constant factor.  This duality maps the \emph{greedy} BST algorithm to the smooth heap.  The greedy BST algorithm has many if not all of the properties of the splay algorithm~\cite{DHIKP09,Fox11,Luc88,mehlhorn2015greedy,Mun00}, but unlike splay it is not practical to implement.  On the other hand, smooth heaps are simple, as we shall see.  Some results about the greedy BST algorithm transfer to smooth heaps through the duality, but the correspondence is limited to the \emph{sorting mode}, in which the heap operations are a sequence of insertions followed by a sequence of delete-mins.

The slim heap, a close variant of the smooth heap, was introduced by Hartmann, Kozma, Sinnamon, and Tarjan~\cite{HKST21}.  They extended smooth heaps and slim heaps to support the entire repertoire of heap operations, and they did a direct analysis of both data structures, obtaining amortized bounds of $\OO(\lg n)$ for delete-min and decrease-key and $\OO(1)$ for the other heap operations.  They also showed that Elmasry's mechanism can be used in both data structures to reduce the amortized time for decrease-key to $\OO(\lg\lg n)$.  Like pairing heaps, the implementations that use Elmasry's idea are not subject to Iacono and {\"O}zkan's lower bound.

We show that both slim heaps and smooth heaps without the Elmasry mechanism answer our question in the affirmative: Both take $\OO(\lg n)$ time per delete-min, $\OO(\lg\lg n)$ time per decrease-key, and $\OO(1)$ time for each other heap operation.  These data structures are subject to Iacono and {\"O}zkan's lower bound, and are the first data structures in their model to match their lower bounds for all operations.

Our analytical tools are interesting in their own right.  We adapt and extend an idea of Iacono~\cite{IaconoPairing} that allows us to reduce the amortized time per insert and meld to $\OO(1)$.  More significantly, we use a potential function based on one used in Pettie's analysis of pairing heaps~\cite{PettiePairing}, but modified to grow much more slowly.  The use of this potential function in a novel analysis allows us to significantly reduce the amortized time bound for decrease-key.

Table~\ref{Ta:known-bounds} contains a summary of known bounds for pairing heaps, multipass pairing heaps, slim heaps, and smooth heaps, including our new results.

\begin{table}
\centering
\setlength\tabcolsep{1.6pt}
\def\arraystretch{2}
\begin{tabular}{|c|c|c|c|}
\hline
\textbf{Heap} & delete-min/delete & decrease-key & insert/meld \\
\hline\hline
Fibonacci~\cite{Fibonacci} & $\OO(\lg n)$ & $\OO(1)$ & $\OO(1)$\\
\hline
Pairing~\cite{FSST86} & $\OO(\lg n)$ & $\OO(\lg n)$ & $\OO(\lg n)$\\
\hline
Pairing~\cite{IaconoPairing} & $\OO(\lg n)$ & $\OO(\lg n)$ & $\OO(1)$\\
\hline
\makecell{Pairing~\cite{Elmasry09, Elmasry17}\\ (variant decrease-key)} &  $\OO(\lg n)$ & $\OO(\lg\lg n)$ & $\OO(1)$\big/$\OO(\lg\lg n)$\\
\hline
Pairing~\cite{PettiePairing} & $\OO(\lg n)$ & $\OO(2^{2\sqrt{\lg\lg N}})$ & $\OO(2^{2\sqrt{\lg\lg N}})$\\
\hline
Pairing~\cite{phd2022} & $\OO(\lg n)$ & $\OO(\sqrt{\lg\lg N}2^{\sqrt{2\lg\lg N}})$ & $\OO(\sqrt{\lg\lg N}2^{\sqrt{2\lg\lg N}})$\\
\hline
Multipass pairing~\cite{FSST86} & $\OO\Bigr(\frac{\lg n \cdot \lg\lg n}{\lg\lg\lg n}\Bigr)$ & $\OO\Bigr(\frac{\lg n}{\lg\lg\lg n}\Bigr)$ & $\OO\Bigr(\frac{\lg n}{\lg\lg\lg n}\Bigr)$\\
\hline
Multipass pairing~\cite{DKKPZ} & $\OO(\lg n \cdot \log^*n \cdot 2^{\log^*N})$ & $\OO\Bigr(\frac{\lg n \cdot 2^{\log^*N}}{\lg^2\lg n}\Bigr)$ & $\OO\Bigr(\frac{\lg n \cdot 2^{\log^*N}}{\lg^2\lg n}\Bigr)$\\
\hline
\makecell{Multipass pairing\\\relax [this work]} & $\OO(\lg n)$ &  $\OO(\lg\lg n \cdot \lg\lg\lg n)$ & $\OO(1)$\\
\hline
Slim \& Smooth~\cite{HKST21} &  $\OO(\lg n)$ & $\OO(\lg n)$ & $\OO(1)$\\
\hline
\makecell{Slim \& Smooth~\cite{HKST21}\\ (variant decrease-key)} &  $\OO(\lg n)$ & $\OO(\lg\lg n)$ & $\OO(1)$\\
\hline
\makecell{Slim \& Smooth\\\relax [this work]} & $\OO(\lg n)$ & $\OO(\lg\lg n)$ & $\OO(1)$\\
\hline
\end{tabular}
\caption{New and previous bounds for self-adjusting heaps.  The bounds for Fibonacci and pairing heaps are included for comparison.  All bounds are amortized. Here $n$ denotes the number of items in the heap at the time of the operation and $N$ denotes the maximum number of items ever in a heap. For the bounds in this work, $n$ may be decreased to the number of items in the heap that will eventually be deleted. The find-min and make-heap operations take constant time in all cases. The heaps labeled ``variant decrease-key'' use Elmasry's buffered sorting method for decrease-key~\cite{Elmasry09, Elmasry17}. Excluding these heaps and Fibonacci heaps, if insert and delete-min take $\OO(\lg n)$ amortized time in a self-adjusting heap, then decrease-key must take $\Omega(\lg\lg n)$ amortized time~\cite{FredmanLB, IaconoOzkan}.}
\label{Ta:known-bounds}
\end{table}

An additional benefit of our analysis is that in our bounds we can take $n$ to be the number of items in the heap or heaps that are eventually deleted.  That is, items inserted but never deleted do not count, except for a cost of $\OO(1)$ per insertion of such an item. 

Our paper is a combination and significant revision of two conference papers~\cite{SODAMultipass, SODASmoothSlim}.  In addition to this introduction, our paper contains ten sections.  Section~\ref{S:canonical-framework} presents a framework for heap implementations that encompasses all the data structures we consider, and many others.  Section~\ref{S:self-adjusting-heaps} presents pairing heaps, multipass pairing heaps, slim heaps, and smooth heaps as instances of the canonical framework.  Section~\ref{S:terminology} summarizes the terminology and definitions used in our analysis. Section~\ref{S:overview} outlines the proof structure we use in the analysis of each of the heaps. Section~\ref{S:time-shifting-lemmas} contains a few central lemmas that apply to all three heaps. Sections~\ref{S:multipass-analysis}, \ref{S:slim-analysis}, and~\ref{S:smooth-analysis} analyze multipass pairing, slim, and smooth heaps, respectively.  Section~\ref{S:lazy} extends our results to lazy (forest) versions of the three data structures.  Section~\ref{S:remarks} contains some concluding remarks and open problems.

\section{A Canonical Heap Framework}
\label{S:canonical-framework}

All the heap implementations we study, and many others, are instances of a \emph{canonical heap framework}.  This framework represents a heap by a rooted tree whose nodes are the heap items.  That is, the data structure is \emph{endogenous}~\cite{tarjan1983data}.  We denote by $v.\key$ the key of heap node $v$.  The tree is \emph{heap-ordered} by key: If $v$ is the parent of $w$, $v.\key \leq w.\key$.  Thus the root of the tree is a node of minimum key.  Access to the tree is via the root, so a find-min takes $\OO(1)$ time.

Heap-ordered trees are modified by \emph{links} and \emph{cuts}.  A link of the roots of two node-disjoint trees combines the two trees into one by making the root of smaller key the parent of the root of larger key, breaking a tie arbitrarily.  The new parent is the \emph{winner} of the link; the new child is the \emph{loser} of the link.  We denote by $vw$ a link of two roots $v$ and $w$ won by $v$ and lost by $w$.  We use the term ``link" to denote the new parent-child pair $vw$ as well as the operation of making $w$ a child of $v$.

A \emph{cut} of a child $w$ breaks the link between $w$ and its parent, say $v$.  This breaks the tree containing $v$ and $w$ into two trees, one containing $w$ and all its descendants in the original tree, the other containing all nodes in the original tree that are not descendants of $w$.

Given an appropriate tree representation~\cite{FSST86}, a link or cut takes $\OO(1)$ time.

The canonical framework implements the heap operations as follows:

\begin{itemize}
\setlength\itemsep{1em}
\item[] $\Makeheap()$: Create and return a new, empty tree.

\item[] $\Findmin(H)$: If $H$ is nonempty, return its root; otherwise, return null.

\item[] $\Insert(H, e)$: If $H$ is nonempty, link $e$ with the root of $H$; if $H$ is empty, make $e$ the root of $H$.

\item[] $\Meld(H_1, H_2)$: If $H_1$ or $H_2$ is empty, return the other; if both are nonempty, link their roots and return the resulting tree.

\item[] $\Decreasekey(H, e, k)$: Change the key of $e$ to $k$.  If $e$ is not the root of $H$, cut $e$ from its parent and link $e$ with the root of $H$.

\item[] $\Deletemin(H)$: Delete $\Findmin(H)$ from $H$.  This makes each child of the deleted root into a new root.  Repeatedly link two of these new roots until only one remains.

\end{itemize}

The only flexibility in the canonical framework is in the algorithm for linking roots in delete-mins.  The self-adjusting heaps we study determine which links to do by maintaining each set of children as an ordered list.  When discussing such a list, we shall use \emph{left} and \emph{right} to mean ``toward the front of the list" and ``toward the back of the list," respectively.  With the exception of certain links done during delete-min operations in smooth heaps (as we discuss in the next section), each link makes the loser the new \emph{leftmost} child of the winner.  We call such a link a \emph{one-sided} link.  If all links are one-sided, the children of a node are in decreasing order by link time, the leftmost child being the latest.

In a delete-min, deletion of the root makes its list of children into a list of roots.  The heaps we study link only adjacent roots on this list, removing the loser and leaving the winner in its current position on the list.  We call a link $vw$ done during a delete-min a \emph{left} or \emph{right} link, respectively, if the loser $w$ is left or right of the winner $v$ on the root list before the link.  Links done during insertions, melds, and decrease-keys are neither left nor right.

The canonical framework as we have presented it is the \emph{eager} or \emph{tree} framework.  An alternative is the \emph{lazy} or \emph{forest} framework, used for example in the original version of Fibonacci heaps.  In this framework, a heap is not a single tree but a set of trees, with the roots stored in a circular list accessed via a root of minimum key, the \emph{min-root}.  An insertion into a nonempty heap does not do a link but makes the new node into a one-node tree and adds it to the root list.  A meld of two nonempty trees catenates their root lists.  A decrease-key  of a child $w$ cuts $w$ from its parent and adds $w$ to the list of roots.  In each case the operation also updates the min-root.  A delete-min deletes the min-root, catenates its list of children with the root list remaining after deletion of the min-root, and repeatedly links pairs of roots on the root list until only one remains.

The versions of slim and smooth heaps analyzed in~\cite{HKST21} use the lazy framework, but we shall study the eager versions, since their implementation and analysis are simpler.  Our bounds for multipass pairing, slim, and smooth heaps also hold for the lazy versions, as we discuss in Section~\ref{S:lazy}.  In general one can modify a heap implementation in the lazy framework to produce one in the eager framework and vice-versa, although this may require redefining how linking is done in a delete-min, and it affects the efficiency of the heap operations.  For example, there is an eager implementation of Fibonacci heaps~\cite{KaplanTZ14}, but it requires doing ``unbalanced" links as well as balanced ones.  The lower bound models of Fredman~\cite{FredmanLB} and Iacono and {\"O}zkan~\cite{IaconoOzkan} encompass both the eager and the lazy framework but place constraints on the links that are allowed in delete-min operations.

\begin{figure}[ht!]
\centering
\includegraphics[width=4.3in]{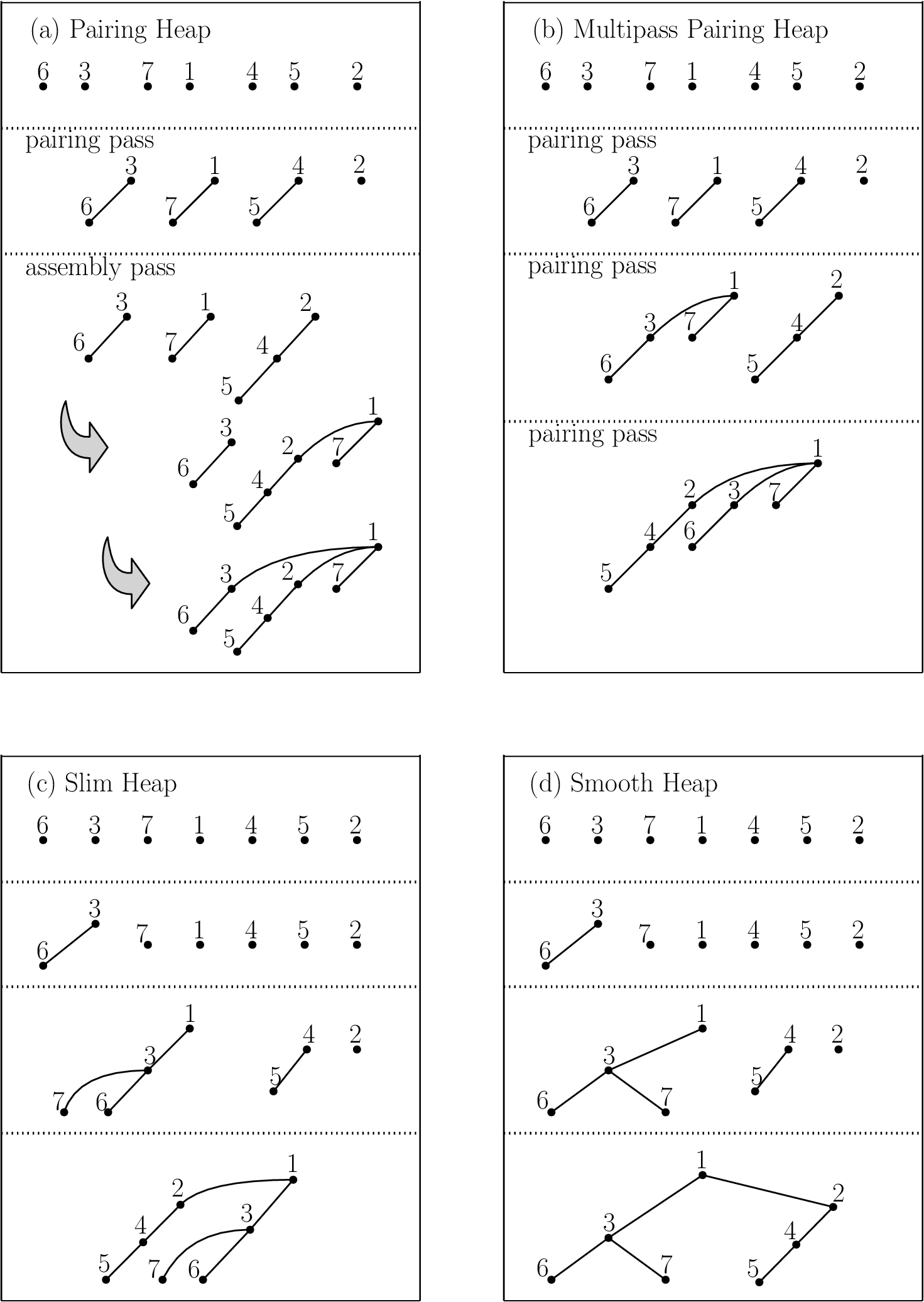}
% \Description{An example of linking during a delete-min for the pairing heap, multipass pairing heap, slim heap, and smooth heap.}
\caption{Linking during delete-min, after the original root is deleted, for the pairing heap (a), multipass pairing heap (b), slim heap (c), and smooth heap (d). The circles represent nodes and the numbers indicate their keys. In (c) and (d), the state of the heap is shown after every two or three links. Smooth heaps do \emph{stable} links (explained in this section), so the losers of links in (d) are positioned on the left or the right of the winner, according to their relative order before the link.}
\label{F:heap-examples}
\end{figure}

\section{Self-Adjusting Heaps}\label{S:self-adjusting-heaps}

The heaps we study are all instances of the canonical framework, differing only in the algorithm by which they do links during a delete-min, and in how the link operation is implemented.  We specify this algorithm for each heap implementation.  We include pairing heaps for reference.    

The \emph{pairing heap} links roots during a delete-min in two passes.  The first pass, \emph{pairing}, links the roots on the root list in adjacent pairs, the first with the second, the third with the fourth, and so on.  If there is an odd number of roots, the rightmost root is not linked during the pairing pass.  The second pass, \emph{assembly}, repeatedly links the rightmost root with its left neighbor until only one root remains.  See Figure~\ref{F:heap-examples}(a).

The \emph{multipass pairing heap} links roots during a delete-min by doing repeated pairing passes until only one root is left. See Figure~\ref{F:heap-examples}(b).

Slim and smooth heaps both link roots during a delete-min by doing \emph{leftmost locally maximum linking}: Find the leftmost three consecutive roots $u$, $v$, $w$ on the root list such that $v.\key \geq \max\{u.\key, w.\key\}$, and link $v$ with whichever of $u$ and $w$ has larger key.  Repeat until there is only one root.  If $u$ and $w$ have the same key, break the tie by linking $v$ with $u$.  As special cases, if $v$ is the leftmost root, $w$ is its right neighbor, and $v.\key \geq w.\key$, link $v$ and $w$; if all the roots are in strictly increasing order by key, link the rightmost two roots.  The two special cases are equivalent to adding leftmost and rightmost dummy roots, each with key $-\infty$, and doing leftmost locally maximum linking until one non-dummy root remains.  See Figure~\ref{F:heap-examples}(c) and (d).

It is straightforward to do leftmost locally maximum linking in $\OO(1)$ time and at most two comparisons per link~\cite{KS19}: After deleting the original root, initialize the current root to be the leftmost root on the root list.  If the current root is not a local maximum, replace the current root by its right neighbor; if it is, link it to its neighbor of larger key and replace the current root by the winner of the link.  Repeat until only one root remains.

\begin{remark}
One can generalize leftmost locally maximum linking to allow \emph{arbitrary} locally maximum links: Given a root whose key is larger than those of its neighbors, link it with whichever neighbor has larger key.  This works as long as ties in keys are broken consistently.  See~\cite{HKST21}.
\end{remark}

The only difference between slim and smooth heaps is in how they maintain the order of children when doing links during delete-min operations. Slim heaps do one-sided links, defined in the previous section: The loser of a link during a delete-min becomes the new leftmost child of the winner.  Smooth heaps do \emph{stable} links: The loser of a link during a delete-min becomes the new leftmost or rightmost child of the winner if the loser was left or right of the winner on the root list, respectively.  See Figure~\ref{F:s-heaps}.  Pairing heaps, multipass pairing heaps, Fibonacci heaps, and many other heap implementations do one-sided links (although in Fibonacci heaps and similar ``balanced" heap implementations, links are not determined by order on the root list but by balance information, so the order of children is irrelevant).  In both slim and smooth heaps, each link done during an insertion, meld, or decrease-key is one-sided.

\begin{figure}[h!]
\centering
\includegraphics[width=4.5in]{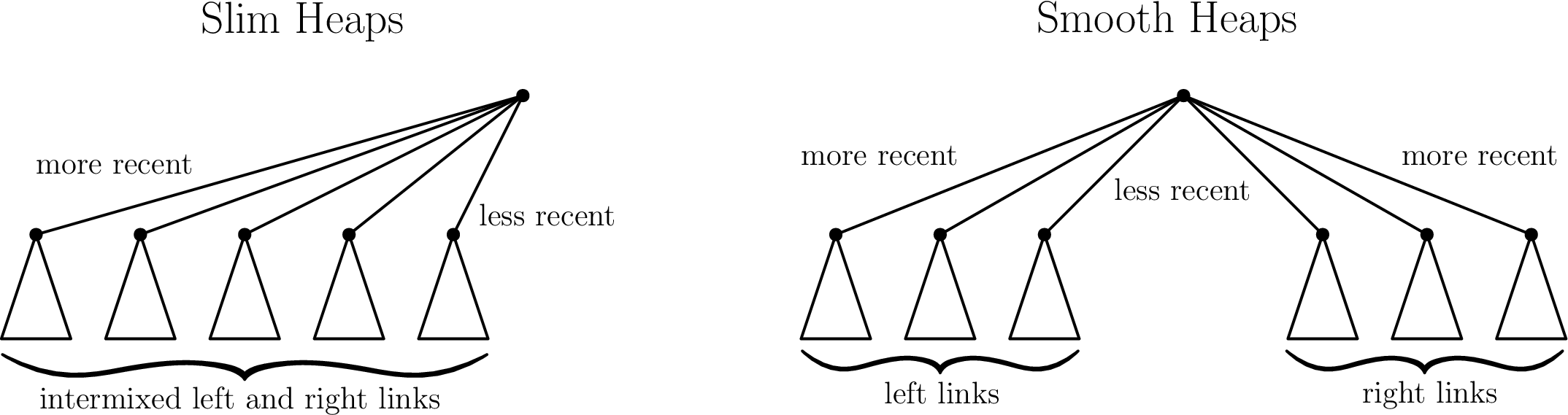}
% \Description{An illustration of links in a slim heap and a smooth heap.}
\caption{Linking in slim and smooth heaps.  In a slim heap, the loser of a link becomes the leftmost child of the winner.  In a smooth heap, the loser becomes the leftmost or rightmost child of the winner, depending on whether the link is left or right.}
\label{F:s-heaps}
\end{figure}

\section{Terminology}\label{S:terminology}

Our presentation of the analysis of these data structures uses quite a bit of terminology, much of which we introduce here.  The following definitions apply to all the heaps we study, except as noted.  We shall repeat these definitions as we use them.

\begin{itemize}
\setlength\itemsep{1em}

\item [] A node is \emph{temporary} if it is eventually deleted, \emph{permanent} otherwise.

\item [] A link is an \emph{insertion link} if it is done during an insertion or meld, a \emph{decrease-key} link if it is done during a decrease-key, a \emph{deletion link} if it is done during a delete-min. 

\item [] A deletion link $vw$ is a \emph{left link} if $w$ precedes $v$ on the root list before the link, a \emph{right link} otherwise.  Insertion and decrease-key links are neither left nor right.

\item [] A link is a \emph{d-link} if it is cut by a delete-min, a \emph{k-link} if it is cut by a decrease-key, an \emph{f-link} if it is never cut. (The \emph{f} stands for ``final''.)

\item [] A link is \emph{real} if its winner and loser are both temporary and it is a d-link, \emph{phantom} otherwise.

\item [] A child node is \emph{a real child} if it is connected to its parent by a real link,  \emph{a phantom child} otherwise.  Every permanent child node is a phantom child.  Each time a temporary node loses a link, it becomes a real or phantom child, depending on whether the link it lost is a d-link or a k-link: A temporary node cannot lose an f-link.

\item [] A link done during pairing pass $i$ of a delete-min in a multipass pairing heap is a \emph{pass-$i$ link}.

\item [] A child node in a smooth heap is a \emph{left child} if the link it lost to its parent is an insertion, decrease-key, or left link, a \emph{right child} if the link it lost to its parent is a right link.

\item [] If $w$ is a temporary node, its \emph{real subtree} is the partial subtree induced by the set of descendants of $w$ (including $w$) that are connected to $w$ by a path of real links.

\item [] If $w$ is a temporary node, its \emph{size} $w.s$ is the number of nodes in its real subtree.  Permanent nodes have no size.  The size of a temporary node changes only when it wins a real link. 

\item [] If $w$ is a temporary node, it has a \emph{mass} $w.m$.  For multipass pairing heaps and slim heaps, we define $w.m$ as follows: If $w$ is a phantom child or a root not in the middle of a delete-min, $w.m = w.s$.  If $w$ is a real child of a temporary node $v$ in a multipass pairing heap or a slim heap, $w.m$ is the size of $v$ just after link $vw$ was done.  If $w$ is a root in the middle of a delete-min in a multipass pairing heap or a slim heap, $w.m$ is $w.s$ plus the sum of the sizes of all the temporary roots to its right on the root list, with certain exceptions depending on the type of heap.  We specify these exceptions in Sections~\ref{S:multipass-analysis} and~\ref{S:slim-analysis}.  We define the mass of temporary nodes in a smooth heap in Section~\ref{S:smooth-analysis}.  Permanent nodes have no mass.  The mass of a temporary node can change, but in a multipass pairing heap or a slim heap the mass of a node changes only when the node wins or loses a real link, or its parent is deleted by a delete-min. In smooth heaps, masses can change for other reasons, as we discuss in Section~\ref{S:smooth-analysis}.

\item [] If $w$ is a temporary node, the \emph{rank} $w.r$ of $w$ is $\lg\lg(w.m/w.s)$. 

\item [] For multipass pairing heaps, if $w$ is a temporary node and $b$ a nonnegative integer, the \emph{$b$-reduced rank} $w.r(b)$ of $w$ is $\max\{w,r - b, 0\} = \lg(\lg(w.m/w.s)/2^b)$.

\item [] If $vw$ is a real link, then $v.r_w$ and $v.r'_w$ and $w.r_v$ and $w.r'_v$ are the rank of $v$ before and after the link, and the rank of $w$ before and after the link, respectively.  This notation does not distinguish between the winner and the loser of the link, but the context will provide the distinction.

\item []  Let $vw$ be a real deletion link. The \emph{winner link} of $vw$ is the last link lost by $v$ before it wins $vw$.  This is the (unique) link $uv$ cut by the delete-min that does $vw$.  Since $u$ and $v$ are both temporary and $uv$ is cut by a delete-min, $uv$ is real.

\item []  Let $vw$ be a real deletion link with winner link $uv$. The \emph{link rank} of $vw$ is $v.r_u - u.r'_v$, the rank of $v$ just before $uv$ is done (which happened sometime before the deletion of $u$) minus the rank of $u$ just after $uv$ is done.\footnote{This definition of link rank differs from the one in our conference papers~\cite{SODAMultipass, SODASmoothSlim}, as does the next definition, of $w$-shift.  The new definitions allow us to simplify the analysis a bit.}

\item [] Let $vw$ be a real deletion link with winner link $uv$. The \emph{$w$-shift} of $vw$ is $v.r_w' - u.r_w'$, the change in the rank of the parent of $w$ resulting from cutting $uw$ and doing $vw$.

\item [] Let $vw$ be a real deletion link with winner link $uv$.  The \emph{$v$-shift} of $vw$ is $v.r_u - v.r_w'$, the net decrease in the rank of $v$ from just before link $uv$ was done until just after $vw$ was done. 

\item [] Let $vw$ be a real deletion link with winner link $uv$. The \emph{$u$-shift} of $vw$ is $u.r_w' - u.r_v'$, the rank of $u$ just after $uw$ was done minus the rank of $u$ just after $uv$ was done.  
\end{itemize}

Let $vw$ be a real deletion link with winner link $uv$.  The following equality, which we call the \emph{link-rank equality}, relates the link rank, the $w$-shift, the $v$-shift, and the $u$-shift of $vw$:
\begin{align*}
v.r_u-u.r_v' = (v.r_w' - u.r_w') + (v.r_u - v.r_w') + (u.r_w' - u.r_v')
\end{align*}

That is, the rank of $vw$ equals its $w$-shift plus its $v$-shift plus its $u$-shift.  The equality follows by canceling equal terms on the right-hand side.

If $f$ is a function of $n$, we use the terminology ``$f(n)$ per delete-min" to mean the sum of $f(n_i)$ over the delete-min operations, where $n_i$ the number of temporary nodes in the heap just before the $i$-th delete-min.  Similarly, ``$f(n)$ per decrease-key" is the sum of $f(n_i)$ over all the decrease-key operations, where $n_i$ is the number of temporary nodes in the heap during the $i$-th decrease-key, and ``$f(n)$ per k-link" is the sum of $f(n_i)$ over all the k-links, where $n_i$ is the number of temporary nodes in the heap at the time the $i$-th k-link was done.

\section{Overview of the Analysis}\label{S:overview}

To bound the total time of a sequence of heap operations, it is enough to count the number of links they do, since each link takes $\OO(1)$ time, and the time spent per operation is $\OO(1)$ plus at most a constant times the number of links it does.  The number of melds that do a link is at most the number of insertions that do not do a link, since each link done by a meld reduces the number of non-empty heaps by one, and the only way to create a non-empty heap is to do an insertion that does not do a link.  It follows that the number of insertion links is at most the number of insertions.  Thus the total number of links is at most one per insertion plus one per decrease-key plus the total number of deletion links.  Hence it suffices to count the total number of deletion links. We denote this number by $d$.

In stating bounds we denote by $n$ the number of temporary nodes in the heap or heaps undergoing an operation.  A benefit of our analysis is that our bounds depend only on $n$, not on the total number of nodes.  That is, nodes that are inserted but never deleted do not count in the bounds.  To simplify the statement of certain bounds we assume $n \geq 8$: Any operation on a heap or heaps with fewer than $8$ nodes takes $\OO(1)$ time, and all sizes, masses, and ranks in such a heap or heaps are bounded by a constant.

In the rest of this section we outline the ideas in our analysis.  The details differ for each type of heap, but the high-level approach is the same.  Our overview is necessarily inaccurate, since each type of heap requires adaptations in the general approach.

There are two parts.  The first, simpler part is to show that it suffices to count real deletion links, or in the case of multipass pairing heaps just a subset of these links.  The second, much more complicated part is to bound the number of such links.  For this we use node and link ranks. Ultimately we charge each real deletion link to some heap operation or to some change or changes in the sum of node ranks.  For this to work we need to bound the sum of the magnitudes (absolute values) of the changes in node ranks.  Obtaining this bound is complicated.  We shall state the bound when we need it but postpone its proof to the end of the analysis.

Central to our analysis are link ranks.  We shall prove that every link rank is non-negative and that the sum of all link ranks is at most our desired bound on real deletion links.  We do the latter indirectly, by bounding the sums of the $w$-shifts, $v$-shifts, and $u$-shifts of the real deletion links and applying the link-rank equality.  The bound on the sum of link ranks implies that the number of real deletion links with large rank (at least $1$) is within our desired bound.

We shall also prove that each link of small rank corresponds either to a node of small rank existing at a certain time, a number we can bound, or to a decrease of at least $1$ in a node rank.  The number of such decreases is within our desired bound, by our bound on the sum of the magnitudes of node-rank changes.

We do these steps in the following order:

\begin{itemize}[align=left]
\setlength\itemsep{1em}
\item[Step 1:] We bound the sum of the $w$-shifts.  This step does not depend on the type of heap.
\item[{Step 2:}] We show that it suffices to count real deletion links, or in the case of multipass pairing heaps to count a subset of these.
\item[Step 3:] We show that all link ranks are non-negative.
\item[Step 4:] We bound the number of real deletion links of small rank.
\item[Step 5:] We bound the sum of link ranks, by separately bounding the sum of the $v$-shifts and the sum of the $u$-shifts.  This step uses our bound on the sum of the magnitudes of node-rank changes, or in the case of multipass pairing heaps a more general bound.
\item [Step 6:] We prove our bound on the sum of the magnitudes of node-rank changes, or in the case of multipass pairing heaps a more general bound.
\end{itemize}

To do Step 6, we observe that node ranks, like link ranks, are always non-negative, and the sum of node ranks is initially zero, since initially there are no nodes.  It follows that the sum of the node-rank \emph{increases} is always at least as large as the sum of the node-rank \emph{decreases}, and hence that the sum of the magnitudes of the node-rank changes is at most twice the sum of the node-rank increases.  We bound the sum of these increases. 

We proceed with the details.  We begin with two results having to do with time-shifting, whose proofs do not depend on the type of heap.  The rest of our analysis depends on the type of heap.  We analyze multipass pairing heaps, slim heaps, and smooth heaps in turn, in Sections~\ref{S:multipass-analysis}, ~\ref{S:slim-analysis}, and~\ref{S:smooth-analysis} respectively.  In our analysis of multipass pairing heaps, we generalize the notion of node rank.  This allows us to reduce the effect of the changes in node ranks caused by pass-$i$ links as $i$ increases.  In our analysis of slim heaps, we extend a different notion, that of size.  This allows us to bound the effect of phantom links.  To analyze smooth heaps, we split the treap formed during a delete-min into three parts: a left side, a right side, and a central path.  We apply the analysis of slim heaps separately to the left side and to the right side, and we apply a new analysis to the central path.

\section{Time-Shifting Lemmas}\label{S:time-shifting-lemmas}

Our first time-shifting lemma allows us to shift the cost of each k-link from the time the k-link is done to the later time when it is cut by a decrease-key:
% C: I reversed this. Previously it said:
% Our first time-shifting lemma allows us to shift the cost of each decrease-key from the time the decrease-key is done to the earlier time when the k-link that it cuts was done:

\begin{lemma}\label{L:k-link-shift}
In any sequence of heap operations starting with no heaps, the quantity $\lg\lg n$ per k-link is at most $\lg\lg n$ per decrease-key plus $3$ per delete-min.
\end{lemma}

\begin{proof}
We use a credit argument.  We allocate $\lg\lg n$ credits to each decrease-key and $3$ to each delete-min and show that the total number of credits is at least $\lg\lg n$ per k-link.

Let $vw$ be a k-link.  Let $n'$ and $n$ be the number of nodes in the heap containing $w$ just after $vw$ is done and when $vw$ is cut, respectively.  We use the $\lg\lg n$ credits allocated to the decrease-key that cuts $vw$ to partially or fully cover the $\lg\lg n'$ credits needed by link $vw$.  This suffices unless $n' > n$, in which case we need an additional $\lg\lg n' - \lg\lg n$ credits.  These we obtain from the credits allocated to the delete-min operations.

A delete-min on a heap of size $n''$ distributes its three credits equally among the nodes in the heap after the delete-min, $3/(n''-1)$ to each.  From the time $vw$ is done until the time $vw$ is cut, $w$ receives at least $3\sum_{i=n}^{n'-1} 1/i \geq 2(\lg n' - \lg n) \geq \lg\lg n' - \lg\lg n$ credits, where the last inequality holds because the derivative of $2\lg x$ is at least that of $\lg\lg x$ for $x \geq 0$.  These credits cover the remainder of the credits needed by link $vw$.
\end{proof}

Our second time-shifting lemma bounds the sum of the $w$-shifts of the real deletion links.  Its proof uses the time-shifting idea in the proof of Lemma~\ref{L:k-link-shift}, an idea that we shall also use later.  Recall from Section~\ref{S:terminology} that if $vw$ is a real deletion link with winner link $uv$, its $w$-shift is $v.r_w' - u.r_w'$, the change in the rank of the parent of $w$ resulting from cutting $uw$ and doing $vw$.

\begin{lemma}\label{L:w-shifts}
In any sequence of operations on multipass pairing heaps, slim heaps, or smooth heaps starting with no heaps, the sum of the $w$-shifts of all real deletion links is at most $\lg\lg n$ per decrease-key plus $\lg\lg n + 3$ per delete-min.  
\end{lemma}

\begin{proof}
Let $w$ be a temporary node.  Consider an interval $I$ from the time $w$ is a root not in the middle of a delete-min until it next becomes a root not in the middle of a delete-min.  We sum the $w$-shifts of the real deletion links lost by $w$ during $I$, and then sum over all $w$, $I$ pairs.

The first link $vw$ lost by $w$ during $I$ is an insertion or k-link, which makes the rank of $v$ zero, since just after the link the size and mass of $v$ are equal.  During the rest of interval $I$, node $w$ can change parent, but only during a delete-min in which it loses a link.  It remains a real child, except temporarily in the middle of a delete-min, until it is the last remaining root in a delete-min, ending $I$, or it loses a phantom deletion link.  Once $w$ loses a phantom link $vw$, it does not change parent until $vw$ is cut by a decrease-key.  This ends $I$, since it makes $w$ a root not in the middle of a delete-min.  Let $uw$, $vw$, and $xw$ be three consecutive links lost by $w$, with $vw$ and $xw$ real deletion links.  Then $uw$ is real and is the winner link of $vw$ and $vw$ is the winner link of $xw$.  The $w$-shifts of $vw$ and $xw$ are $v.r_w' - u.r_w'$ and $x.r_w'-v.r_w'$, respectively.  Their sum telescopes to $x.r_w'-u.r_w'$.  More generally, the sum of the $w$-shifts of the real deletion links lost by $w$ during $I$ telescopes to a difference of two node ranks, whose value is at most $\lg\lg n'$, where $n'$ is the size of the heap containing $w$ when $w$ loses its last real link during $I$.

To finish the proof of the lemma we use the time-shifting idea in the proof of Lemma~\ref{L:k-link-shift}.  We allocate $\lg\lg n$ credits to each decrease-key and $\lg\lg n + 3$ to each delete-min, and show that the total number of credits is at least $\lg\lg n'$ per triple $w$, $I$, $n'$, where $n'$ is the size of the heap containing $w$ when $w$ loses its last real link during $I$.

Given a triple $w$, $I$, $n'$, let $n$ be the size of the heap containing $w$ when $I$ ends.  We use $\lg\lg n$ of the credits allocated to the heap operation that ends $I$ to partially or fully cover the $\lg\lg n'$ credits needed by the triple.  This suffices unless $n' > n$, in which case we need an additional $\lg\lg n' - \lg\lg n$ credits.  These we obtain from the extra $3$ credits allocated to each delete-min operation.

A delete-min on a heap of size $n''$ distributes its extra three credits equally among the nodes in the heap after the delete-min, $3/(n''-1)$ to each.  From the time $w$ loses its last real link during $I$ until the end of $I$, $w$ receives at least $3\sum_{i=n}^{n'-1} 1/i \geq 2(\lg n' - \lg n) \geq \lg\lg n' - \lg\lg n$ credits, which cover the remainder of the credits needed by the triple.
\end{proof}

\section{Analysis of Multipass Pairing Heaps}\label{S:multipass-analysis}

\subsection{Basic results}

Now we turn to the analysis of multipass pairing heaps.  Recall that they do repeated pairing passes to link the roots in a delete-min. We begin with three simple results.  First, during a delete-min, the number of links remaining to be done after pairing pass $i$ is exponentially small in $i$.

\begin{lemma}\label{L:few-pass-i-links}
Suppose a delete-min in a multipass pairing heap does a total of $k$ links.  Then the number of links remaining to be done after the first $i$ passes is at most $k/2^i$.
\end{lemma}

\begin{proof}
Suppose there are $j+1$ roots at the beginning of pairing pass $i$.  Then the delete-min does a total of $j$ more links.  Pass $i$ does $\lceil j/2 \rceil$ of these links, leaving $\lfloor j/2 \rfloor$ still to be done.  The lemma follows by induction on $i$.
\end{proof}

Second, it is enough to bound a subset of the real deletion links.  We call a deletion link \emph{primary} if it is a real pass-2 link whose winner link is a pass-1, pass-2, or pass-3 link.  Let $p$ be the number of primary links.  Recall from Section~\ref{S:overview} that we denote by $d$ the total number of deletion links.

\begin{lemma}\label{L:primary-links}
In any sequence of multipass pairing heap operations starting with no heaps, $d$ is at most $8p$ plus $20$ per insertion plus $16$ per decrease-key. 
\end{lemma}
\begin{proof}
Consider a delete-min that deletes root $u$.  Removal of $u$ makes its list of children into a list of new roots, which are then linked in pairing passes.  Group these new roots into sets of four consecutive nodes, leaving at most three roots ungrouped.  The first pairing pass does two links of pairs of nodes within the group.  The second pairing pass links the winners of these two pass-1 links.  Let $vw$ and $xy$ be the two pass-1 links (won by $v$ and $x$, respectively), and let $vx$ be the pass-2 link (won by $v$). We associate these three links with the group. If $v$ is not the one root that remains at the end of the delete-min, it also loses a link in a later pairing pass.  We associate this link with the group also.

We charge these three or four links to a particular event, as follows.  Recall from Section~\ref{S:terminology} that link $vx$ is real if $v$ and $x$ are both temporary (eventually deleted) and $vx$ is a d-link (cut by a delete-min); otherwise $vx$ is phantom.  Suppose $vx$ is phantom.  If $v$ and $x$ are temporary, $vx$ must be a k-link (cut by a decrease-key).  In this case we charge the three or four links to the decrease-key that cuts $vx$.  The total charge incurred in this way is at most four per decrease-key that cuts a link between two temporary nodes.  Suppose $v$ is permanent: The argument is symmetric if $x$ is permanent.  The link $vw$ must be a k-link or an f-link (never cut), since a permanent node cannot win a d-link. If $vw$ is a k-link, we charge the three or four links to the decrease-key that cuts it.  The total charge incurred is at most four per decrease-key that cuts a link won by a permanent node. If $vw$ is an f-link, then $w$ must be permanent, and we charge the three or four links to the insertion of $w$.  The total charge incurred is at most four per insertion of a permanent node.

The remaining possibility is that $vx$ is real.  Let $uv$ be the winner link of $vx$.
If $uv$ is a pass-1, pass-2, or pass-3 deletion link, then $vx$ is primary and we charge the three or four links to $vx$. The total charge incurred is at most $4$ per primary link. Otherwise, we charge the three or four links to $uv$, which is either an insertion link, a decrease-key link, or a deletion link that occurs after pass 3 of a delete-min. Since $v$ only wins one pass-2 link (namely $vx$), $uv$ is charged only once in this way. The total charge is at most $4$ per insertion of a temporary node plus $4$ per decrease-key plus $4d/8$, since the number of deletion links occurring after pass 3 of a delete-min is at most $d/8$ by Lemma~\ref{L:few-pass-i-links}. 

In each delete-min, there are at most two links that are not associated with a group.  We charge such non-associated links to the insertion of the node deleted by the delete-min, which is temporary.  The total charge incurred in this way is at most two per insertion of a temporary node.

Adding our bounds on charges, we find that the total charge, which equals $d$, is at most $10$ per insertion plus $8$ per decrease-key plus $4$ per primary link plus $d/2$.  The lemma follows.
\end{proof}

Our bounds on primary links will include small additive terms in $d$, the total number of deletion links.  To handle such terms, we restate Lemma~\ref{L:primary-links} as follows:

\begin{corollary}\label{C:primary-links-slack}
In any sequence of multipass pairing heap operations starting with no heaps, $d$ is at most $(64/3)\max\{0, p-3d/64\}$ plus $80$ per insertion plus $64$ per decrease-key. 
\end{corollary}
\begin{proof}
We consider two cases.  If $p < 3d/32$, then by Lemma~\ref{L:primary-links} $d$ is at most $3d/4$ plus $20$ per insertion plus $16$ per decrease-key, which implies that $d$ is at most $80$ per insertion plus $64$ per decrease-key.  If on the other hand $p \geq 3d/32$, then $(64/3)p \geq 2d$, which implies $d \leq (64/3)(p-3d/64)$.  Adding the bounds for the two cases gives the corollary. 
\end{proof}

\begin{remark}
This proof and subsequent ones do not minimize the constant factors, they only establish that appropriate constant factors exist.
\end{remark}

Our third and final result of this section is that link ranks are non-negative.  Let us recall the definitions from Section~\ref{S:terminology} that are relevant to this result, and in the process provide the complete definition of masses in a multipass pairing heap.

Let $w$ be a temporary node.  The \emph{real subtree} of $w$ is the partial subtree induced by the set of descendants of $w$ (including $w$) that are connected to $w$ by a path of real links.  The \emph{size} $w.s$ of $w$ is the number of nodes in its real subtree.

The \emph{mass} $w.m$ of $w$ depends on whether $w$ is a child or a root, and if a root whether it is a root in the middle of a deletion.  If $w$ is a phantom child or a root not in the middle of a deletion, its mass is its size.  If $w$ is a real child of a node $v$, its mass is the size of $v$ just after link $vw$ was done.  This is the sum of the sizes of $v$ and $w$ just before the link was done.  If $w$ is a root in the middle of a delete-min, its mass is the sum of its size and those of all the temporary roots to its right on the root list, \emph{with the following exception}: Just before a root $w$ loses a real right link to the root $v$ to its left on the root list, the mass of $w$ \emph{steps up} to the mass of $v$ (which is the same before and after the link).  

Finally, the \emph{rank} $w.r$ of $w$ is $\lg\lg(w.m/w.s)$.

Let $vw$ be a real deletion link done during the delete-min that deletes root $u$.  This delete-min cuts the winner link $uv$ of $vw$.  The link rank of $vw$ is $v.r_u - u.r'_v$, the rank of $v$ just before $uv$ is done minus the rank of $u$ just after $uv$ is done.  If $uv$ is a right link, $v.r_u$ is the rank of $v$ after the step-up in the mass of $v$ that occurs just before the link.  

Step-ups guarantee that link ranks are non-negative:

\begin{lemma}
\label{L:multipass-link-ranks-nonnegative}
In a multipass pairing heap, every link rank is non-negative.
\end{lemma}

\begin{proof}
Let $vw$ be a real pairing link with winner link $uv$.  The link rank of $vw$ is $v.r_u - u.r_v'$.  If $uv$ is an insertion or decrease-key link, $u$ has rank $0$ after the link, so $v.r_u - u.r_v'\geq 0$.  Suppose $uv$ is a deletion link. Let unprimed and primed sizes and masses denote values just before and just after the link, respectively. If $uv$ is a left link, $v.m = u.m'$ and $v.s \leq u.s'$, so $v.r_u - u.r_v' = \lg\lg(v.m/v.s) - \lg\lg(u.m'/u.s') \geq 0$.  If $uv$ is a right link, $u.m =u.m'$ and $u.s' = u.s + v.s$.  The step-up just before the link makes the mass of $v$ equal to $u.m'$, so $v.r_u - u.r_v' =\lg\lg(u.m'/v.s)- \lg\lg(u.m'/(u.s+v.s)) \geq 0$.
\end{proof}

\subsection{Small-rank links}\label{S:primary-small}

The next step in our analysis is to bound the number of primary links of small rank.  We begin by deriving a bound on the number of roots of small rank at certain times.

\begin{lemma}
\label{L:multipass-small-node-ranks}
At the beginning of any pairing pass in a delete-min in a multipass pairing heap, the number of roots with rank at most $1$ is at most $(5/2)\lg n +1$. 
\end{lemma}
\begin{proof}
Consider the beginning of some pairing pass, and let $v$ and $w$ be temporary roots with $w$ the next temporary root right of $v$.  Then $v.m-v.s = w.m$.  If $v.r \leq 1$, $v.m/v.s \leq 4$, so $v.s \geq v.m/4$, which implies $v.m \geq (4/3)w.m$.  It follows that if there are $i$ roots on the root list with rank at most $1$, $(4/3)^{i-1} \leq n$.  Hence $i \leq (\lg n)/\lg(4/3) + 1 \leq (5/2)\lg n + 1$.  
\end{proof}

\begin{corollary}
\label{C:multipass-small-rank-winners}
In any sequence of multipass pairing heap operations starting with no heaps, for any fixed $i$ the number of real pass-$i$ links in which the winner or loser has rank at most $1$ just before the link is at most $(5/2)\lg n +1$ per delete-min.
\end{corollary}
\begin{proof}
During a pairing pass, the rank of a node does not decrease until it participates in a pairing link.  Hence a node with rank at most one just before it wins or loses a pairing link had rank at most one at the beginning of the pairing pass.  The corollary is immediate from Lemma~\ref{L:multipass-small-node-ranks}.  
\end{proof}

Next we give a condition that allows us to charge certain real deletion links of rank less than $1$ to decreases in node ranks.

\begin{lemma}
\label{L:multipass-small-link-rank-decreases}
In a multipass pairing heap, let $vw$ be a real deletion link with winner link $uv$.  Suppose the link rank of $vw$ is less than $1$.  If both $u$ and $v$ have positive rank just after $uv$ is done, then $v.r' \leq v.r - 1$, where unprimed and primed variables take their values just before $uv$ is done (after the step-up if there is one) and just after $uv$ is done, respectively. 
\end{lemma}
\begin{proof}
Since $u$ and $v$ have positive ranks just after they are linked, $uv$ is a deletion link.  Furthermore, since link $uv$ (after the step-up in the mass of $v$ if there is one) does not increase the rank of $v$, we can replace $\lg$ by $\log_2$ in the formulas for $u.r'$, $v.r$, and $v.r'$.  Since $v.r - u.r' < 1$, $2^{u.r'} > 2^{v.r-1} = 2^{v.r}/2$.  Suppose $uv$ is a left link.  Then \[2^{v.r'} =\log_2((u.s+v.s)/v.s)=\log_2(v.m/v.s)-\log_2((v.m/(u.s+v.s))=2^{v.r}-2^{u.r'} \leq 2^{v.r}/2\]  Hence $v.r' \leq v.r - 1$.  Suppose $uv$ is a right link.  Then \[2^{v.r'} = \log_2((u.s+ v.s)/v.s) = \log_2(u.m'/v.s) - \log_2(u.m'/(u.s+v.s)) = 2^{v.r}-2^{u.r'} \leq 2^{v.r}/2\]  Hence $v.r' \leq v.r - 1$ in this case also.  
\end{proof}

Corollary~\ref{C:multipass-small-rank-winners} and Lemma~\ref{L:multipass-small-link-rank-decreases} combine to relate the number of primary links of rank less than $1$ to the sum of the magnitudes of all node rank decreases that occur during pairing passes $1$ through $3$ of some delete-min.

\begin{lemma}\label{L:primary-small}
In any sequence of multipass pairing heap operations starting with no heaps, the number of primary links of rank less than $1$ is at most the sum of the magnitudes of all node-rank decreases plus $(15/2)\lg n + 3$ per delete-min. 
\end{lemma}
\begin{proof}
Let $vw$ with winner link $uv$ be a primary link with link rank less than $1$.  Since $v$ wins only one pass-2 pairing link during the deletion of $u$, namely $vw$, $vw$ is uniquely determined by $uv$.  Since $vw$ is primary, $uv$ is a pass-1, pass-2, or pass-3 link.  
If either $u$ or $v$ has rank less than $1$ just before $uv$ is done, we charge $vw$ to $uv$.  By Corollary~\ref{C:multipass-small-rank-winners}, the total of such charges is at most $(5/2)\lg n + 1$ per delete-min for each of pairing passes $1$ through $3$, for a total of $(15/2)\lg n + 3$ per delete-min.   

The remaining possibility is that both $u$ and $v$ have rank at least $1$ just before $uv$ is done.  In this case doing $uv$ reduces the rank of $v$ by at least $1$, by Lemma~\ref{L:multipass-small-link-rank-decreases}.  We charge $vw$ to this node rank decrease.

Adding these bounds gives the lemma.
\end{proof}

To obtain an absolute bound on the number of primary links of rank less than $1$, we need a bound on the sum of the node-rank decreases.  We state a bound on the sum of the magnitudes of the node-rank changes.  We shall give a more general bound in Section~\ref{S:primary-large} and prove it in Section~\ref{S:node-rank-changes}.

\begin{theorem}\label{T:multipass-node-rank-changes}
In any sequence of multipass pairing heap operations starting with no heaps, the sum of the magnitudes of all node-rank changes is $\OO(\lg\lg n)$ per decrease-key plus $\OO(\lg n)$ per delete-min plus at most $d/64$.  
\end{theorem}

Lemma~\ref{L:primary-small} and Theorem~\ref{T:multipass-node-rank-changes} combine to give us a bound on the number of primary links with rank less than $1$:

\begin{theorem}\label{T:multipass-small}
In any sequence of multipass pairing heap operations starting with no heaps, the total number of primary links with rank less than $1$ is $\OO(1)$ per insertion plus $\OO(\lg\lg n)$ per decrease-key plus $\OO(\lg n)$ per delete-min plus at most $d/64$. 
\end{theorem}

\subsection{Large-rank links}\label{S:primary-large}

Now we turn to the task of bounding the number of primary links of large rank.  To do this we bound the sum of link ranks.  We bound this sum indirectly, by using the link-rank equality, which states that if $vw$ is a real deletion link with winner link $uv$, then
\begin{align*}
v.r_u-u.r_v' = (v.r_w' - u.r_w') + (v.r_u - v.r_w') + (u.r_w' - u.r_v')
\end{align*}

That is, the rank of $vw$ equals the sum of its $w$-shift, its $v$-shift, and its $u$-shift.  Summing this equality over all real deletion links $vw$, we find that the sum of the link ranks equals the sum of their $w$-shifts plus their $v$-shifts plus their $u$ shifts.  We already have a bound on the sum of the $w$-shifts, Lemma~\ref{L:w-shifts}.  The proof of Lemma~\ref{L:w-shifts} relies on cancellation in telescoping sums, and the proof fails if we only sum over the primary links, which are the ones we need to count.  Thus we are forced to estimate the sum of the $v$-shifts and the $u$-shifts of \emph{all} the real deletion links, not just those of the primary links.

To bound the sum of these shifts, we observe that a given vertex $v$ can win or lose at most one pass-$i$ link during a single delete-min.  This allows us to bound the sum of shifts over all real pass-$i$ links separately for each $i$.  Then we sum over $i$.  To make the overall bound small, we estimate the sum differently for each $i$.  We use the fact that the number of pass-$i$ links is exponentially small in $i$ (Lemma~\ref{L:few-pass-i-links}).  To exploit this fact, we generalize the definition of ranks.  If $w$ is a temporary node and $b$ is a non-negative integer, the \emph{$b$-reduced rank} $w.r(b)$ of $w$ is $\max\{w.r - b, 0\} = \lg(\lg(w.m/w.s)/2^b)$.  In particular, the $0$-reduced rank of $w$ is just its rank.

The following lemma gives a simple inequality that allows us to relate differences in node ranks to the corresponding differences in $b$-reduced node ranks:

\begin{lemma}\label{L:inequality-1}
If $x > y \geq 0$ and $b \geq 0$, then $x-y \leq \max\{x-b,0\}-\max\{y-b,0\}+b$.
\end{lemma}
\begin{proof}
If $y \geq b$, the inequality becomes $x-y\leq x-y+b$, which is true since $b \geq 0$.  If $x\leq b$, the inequality becomes $x-y\leq b$, which is true since $x\leq b$ and $y\geq 0$.  Finally, if $y<b<x$, the inequality becomes $x-y\leq x$, which is true since $y \geq 0$.
\end{proof}

To bound the sums of the $v$-shifts and the $u$-shifts, we need a generalization of Theorem~\ref{T:multipass-node-rank-changes} that bounds the sum of the magnitudes of the changes in $b$-reduced ranks.  We state and use this theorem here and prove it in the next section.

\begin{theorem}\label{T:b-rank-variation}
Let $b$ and $c$ be non-negative integers.  In any sequence of multipass pairing heap operations starting with no heaps, the sum of the magnitudes of the changes in $b$-reduced node ranks is at most $4\lg((\lg n)/2^b)$ per decrease-key plus $(4(b+c+2)/2^b)\lg n + 2\lg((\lg n)/2^b) + 6/2^b$ per delete-min plus $8d/2^{b+c}$.   
\end{theorem}

Choosing $b=0$ and $c=9$ in Theorem~\ref{T:b-rank-variation} gives Theorem~\ref{T:multipass-node-rank-changes}, which was stated and used in the previous section.

Now we have the tools we need to bound the sums of $v$-shifts and $u$-shifts.  We start with the $v$-shifts.

\begin{lemma}\label{L:multipass-v-shifts}
In any sequence of multipass pairing heap operations starting with no heaps, the sum of the $v$-shifts of the real deletion links is $\OO(\lg\lg n\cdot\lg\lg\lg n)$ per decrease-key plus $\OO(\lg n)$ per delete-min plus at most $d/64$. 
\end{lemma}
\begin{proof}
Let $vw$ be a real pairing link with winner link $uv$.  The $v$-shift of $vw$ is $v.r_u - v.r_w'$, the net decrease in the rank of $v$ from the time $t_1$ just before link $uv$ is done (after the step-up in the mass of $v$ if it is a right link) until the time $t_2$ just after $vw$ is done.  We call $[t_1, t_2]$ the \emph{$v$-shift time interval} of $vw$.  If $vw$ is a pass-$i$ link, $v$ cannot win another pass-$i$ link until after losing a link.  It follows that for a given $v$ and $i$, the $v$-shift time intervals of the pass-$i$ links $vw$ won by $v$ are disjoint.  We can thus apply Theorem~\ref{T:b-rank-variation} to bound the sum of the $v$-shifts of the pass-$i$ links $vw$ for any given $i$.  Summing the resulting bounds over all $i$ gives us a bound on the sum of all $v$-shifts.

In the sum of $v$ shifts, we ignore all the non-positive terms.
If $vw$ is an pass-$i$ link with positive $v$-shift, its $v$ shift $v.r_u-v.r_w'$ is at most $v.r_u(b) - v.r_w'(b) + b$ for any integer $b \geq 0$ by Lemma~\ref{L:inequality-1}, where $v.r_u(b)$ the $b$-reduced rank of $v$ just before $uv$ is done (after the step-up in the mass of $v$ if it is a right link), and $v.r_w'(b)$ is the $b$-reduced rank of $vw$ just after $vw$ is done. 
The number of pass-$i$ links is at most $2d/2^i$ by Lemma~\ref{L:few-pass-i-links}.  Thus, for any integers $i$ and $b$, the sum of $v$-shifts over all real pass-$i$ links is at most the bound of Theorem~\ref{T:b-rank-variation} plus $b \cdot (2d/2^i)$. This bound depends on a non-negative integer parameter $c$, whose value we choose at the end of this proof.

For $i = 1,\dots, c$, this bound with $b = 0$ implies that the sum of the $v$-shifts over all real deletion links done in the first $c$ pairing passes is at most 
$4c\lg\lg n$ per decrease-key plus $4c(c+2)\lg n + 2c\lg\lg n + 6c$ per delete-min plus $8cd/2^c$.

For $i \geq 1$, we apply the same bound to the pass-$(i + c)$ links, but with $b$ a function of $i$, specifically $b = b_i = 2^{\lceil i/2 \rceil}$. By this bound, the sum of the $v$-terms of the real pass-$(i+c)$ links is at most 
\begin{align*}
&\text{$4\lg((\lg n)/2^{b_i})$ per decrease-key}\\
+ &\text{$4((b_i + c + 2)/2^{b_i})\lg n + 2\lg((\lg n)/2^{b_i}) + 6/2^{b_i}$ per delete-min}\\
+ &8d/2^{b_i+c} + 2b_i d/2^{i+c}
\end{align*}
We need to sum this bound from $i=1$ to $\infty$.  The sum of $b_i/2^{b_i}$ is a constant, as are the sums of $1/2^{b_i}$ and $b_i/2^i$.  Finally, if $i > 2\lg\lg\lg n$, the $\lg((\lg n)/2^{b_i})$ terms are $0$.  It follows that if we choose $c$ to be a sufficiently large constant, the sum of the $v$-shifts of real deletion links is within the bound of the lemma.
\end{proof}

The argument for the $v$-shifts applies to the $u$-shifts as well, but we must consider both increases and decreases of node ranks over the $u$-shift time intervals.

\begin{lemma}\label{L:multipass-u-shifts}
In any sequence of multipass pairing heap operations starting with no heaps, the sum of the $u$-shifts of the real deletion links is $\OO(\lg\lg n\cdot\lg\lg\lg n)$ per decrease-key plus $\OO(\lg n)$ per delete-min plus at most $d/64$. 
\end{lemma}
\begin{proof}
Let $vw$ be a real pairing link with winner link $uv$.  The $u$-shift of $vw$ is $u.r_w' - u.r_v'$.  If $vw$ is a right link, this is the decrease in the rank of $u$ from the time $t_1$ just after it wins $uw$ until the time $t_2$ just after it wins $uv$.  If $vw$ is a left link, its $u$-shift is the increase in the rank of $u$ from the time $t_1$ just after it wins $uv$ until the time $t_2$ just after it wins $uw$.  We call $[t_1, t_2]$  the \emph{$u$-shift time interval} of $vw$.    To cover both the left links and the right links, we bound the sum of the magnitudes of the $u$-shifts.  Since each pass-$i$ link is between adjacent roots on the root list, the $u$-shift time intervals of two pass-$i$ links during the deletion of $u$ are disjoint.  These are the only $u$-shifts of pass-$i$ links.  Hence we can apply the same analysis as in the proof of Lemma~\ref{L:multipass-v-shifts}.  
\end{proof}

Summing the link-rank equality over the real deletion links and adding the bounds in Lemma~\ref{L:w-shifts}, Lemma~\ref{L:multipass-v-shifts}, and Lemma~\ref{L:multipass-u-shifts} gives us a bound of the sum of link ranks:

\begin{theorem}\label{T:multipass-link-rank-sum}
In any sequence of multipass pairing heap operations starting with no heaps, the sum of the link ranks of the real deletion links is $\OO(\lg\lg n\cdot\lg\lg\lg n)$ per decrease-key plus $\OO(\lg n)$ per delete-min plus at most $d/32$. 
\end{theorem}

Since all link ranks are non-negative by Lemma~\ref{L:multipass-link-ranks-nonnegative}, the bound in Theorem~\ref{T:multipass-link-rank-sum} is also a bound on the number of primary links with rank at least $1$.  Adding this bound to the bound in Theorem~\ref{T:multipass-small} on the number of such links that have rank less than $1$ and combining the resulting bound with the bound in Corollary~\ref{C:primary-links-slack} gives us our desired bound on the total number of deletion links:

\begin{theorem}\label{T:multipass-deletion-link-bound}
In any sequence of multipass pairing heap operations starting with no heaps, the total number of deletion links is $\OO(1)$ per insertion plus $\OO(\lg\lg n\cdot\lg\lg\lg n)$ per decrease-key plus $\OO(\lg n)$ per delete-min. 
\end{theorem}

Theorem~\ref{T:multipass-deletion-link-bound} yields our bound on the efficiency of multipass pairing heaps:

\begin{theorem}\label{T:multipass-time-bound}
\label{T:multipass}
Any sequence of multipass pairing heap operations starting with no heaps takes $\OO(\lg\lg n \cdot \lg\lg\lg n)$ amortized time per decrease-key, $\OO(\lg n)$ amortized time per deletion, and $\OO(1)$ amortized (and worst-case) time for each other heap operation. 
\end{theorem}

\subsection{Node-rank changes}\label{S:node-rank-changes}

Our analysis of multipass pairing heaps relies on the validity of Theorem~\ref{T:b-rank-variation}, which bounds the sum of the magnitudes of changes in $b$-reduced node ranks.  Proving this theorem is the final step in our analysis and is the topic of this section.

Our calculations require the inequality stated in the following lemma:

\begin{lemma}\label{L:inequality-2}
If $x > y > 0$, $a \geq 0$, and $b \geq 0$, $\lg(x/2^b) - \lg(y/2^b) \leq (2/2^{a+b})(x-y)+a$. 
\end{lemma}

\begin{proof}
Define $f(t) = \lg(t/2^b)$.  Observe that $f'(t) = 0$ for $t < 2^b$ and $f'(t) = \lg(e)/t \leq 2/t$ for $t > 2^b$. In particular, if $t > 2^{a+b}$ then $f'(t) \leq \lg(e)/2^{a+b} \leq 2/2^{a+b}$.  

If $x \geq y \geq 2^{a+b}$, then
\[\lg(x/2^b) - \lg(y/2^b) = \int_y^x f'(t)~dt \leq \Bigr(\frac{2}{2^{a+b}}\Bigr)~\bigr(x - y\bigr)\]
since $f'$ is at most $2/2^{a+b}$ on $(2^{a+b}, \infty)$.
 
If $2^{a+b} \geq x \geq y$, then 
\[\lg(x/2^b) - \lg(y/2^b) \leq \lg(2^{a+b}/2^b) - 0 = a\]

Finally, if $x \geq 2^{a+b} \geq y$, then
\begin{align*}
\lg(x/2^b) &\leq \lg(x/2^b) - \lg(2^{a+b}/2^b) + \lg(2^{a+b}/2^b)\\
&= \Bigr(\int_{2^{a+b}}^x f'(t)~dt\Bigr) + a \tag{since $\lg(2^{a+b}/2^b) = a$}\\
&\leq \Bigr(\frac{2}{2^{a+b}}\Bigr)~\bigr(x - 2^{a+b}\bigr) + a\tag{since $f'(t) \leq 2/2^{a+b}$ on $(2^{a+b}, \infty)$}\\
&\leq \Bigr(\frac{2}{2^{a+b}}\Bigr)~\bigr(x - y\bigr) + a \qedhere
\end{align*} 
\end{proof}

As observed in Section~\ref{S:overview}, the sum of the magnitudes of node-rank changes is at most twice the sum of node-rank increases, since node ranks are always non-negative and their sum is initially zero since initially there are no nodes.  This is also true of $b$-reduced ranks.  Hence to bound the sum of the magnitudes of $b$-reduced rank changes it suffices to bound the sum of $b$-reduced rank increases.  This we proceed to do. 

Newly inserted temporary nodes and nodes deleted by delete-mins have $b$-reduced rank zero, so their insertion or deletion changes no $b$-reduced rank.  A cut of a phantom link, which happens only during a decrease-key, changes no size or mass and hence changes no $b$-reduced rank. Cuts of real links happen only during delete-mins.  Deletion of the root in a delete-min decreases by one the size of each child of the deleted root, which does not increase its $b$-reduced rank, and makes the $b$-reduced rank of the deleted root equal to zero, which does not increase it.  A phantom insertion or decrease-key link changes no size or mass and hence changes no $b$-reduced rank.  A phantom deletion link changes no size and increases no mass, so it increases no $b$-reduced rank.  It remains to study the effect of real links.

\begin{lemma}\label{L:multipass-real-decrease-key-links}
In a multipass pairing heap, a real decrease-key link $vw$ increases the $b$-reduced rank only of $w$, by at most $\lg((\lg n)/2^b)$.  
\end{lemma}

\begin{proof}
The link changes the size only of $v$ and changes the mass only of $v$ and $w$.  Before and after the link, the mass of $v$ equals its size, so the link leaves the $b$-reduced rank of $v$ unchanged at zero.  The link increases the mass of $w$ to at most $n$.  The lemma follows. 
\end{proof}

\begin{lemma}
\label{L:multipass-real-insertion-links}
In any sequence of multipass pairing heap operations starting with no heaps, the sum of increases in $b$-reduced ranks caused by real insertion links is at most $\lg((\lg n)/2^b)$ per decrease-key plus $\lg((\lg n)/2^b) + 3/2^b$ per delete-min.
\end{lemma}

\begin{proof}
The proof uses the idea in the proofs of the lemmas in Section~\ref{S:time-shifting-lemmas}.  We allocate $\lg((\lg n)/2^b)$ credits to each decrease-key and $\lg((\lg n)/2^b) + 3/2^b$ credits to each delete-min, and show that the sum of increases in $b$-reduced ranks caused by real insertion links is at most the sum of the allocated credits.

A real insertion link $vw$ increases the $b$-reduced rank only of $w$, by at most $\lg((\lg n)/2^b)$.  Since $w$ is temporary, it eventually becomes the only root in its heap, which must happen before $w$ can be deleted.  Let $n'$ be the number of temporary nodes in the heap containing $w$ the first time after link $vw$ that $w$ becomes the only root in its heap.  The operation that causes $w$ to become the only root in its heap is either a decrease-key or a delete-min.  

We assign $\lg((\lg n')/2^b)$ of the credits allocated to this decrease-key or delete-min to partially cover the increase in the $b$-reduced rank of $w$ caused by link $vw$.  The credits allocated to a decrease-key or delete-min are assigned to at most one such increase, since $w$ is the only root in the heap after the decrease-key or delete-min, and $vw$ is the most recent previous real insertion link lost by $w$.

If $n' \geq n$, the assigned credits fully cover the increase in the $b$-rank of $w$.  If not, we need an additional $\lg((\lg n)/2^b) - \lg((\lg n')/2^b) \leq (2/2^b)(\lg n - \lg n')$ credits to fully cover the increase.  For these we use the additional $3/2^b$ credits allocated to each delete-min.  Each delete-min distributes its extra credits equally among the $n'' - 1$ temporary nodes in the heap just after the deletion, $(3/2^b)/(n''-1)$ to each.  Between the time $v$ and $w$ are linked and $w$ next becomes the only root in its heap, $w$ receives at least $ (3/2^b)\sum_{i=n'}^{n-1} 1/i \geq (2/2^b)(\lg n - \lg n')$ credits, which cover the remainder of the increase in the $b$-rank of $w$.
\end{proof}

\begin{lemma}\label{L:multipass-real-deletion-links}
In a multipass pairing heap, a real deletion link $vw$ increases a $b$-reduced rank only if it is a right link.  If it is a right link, it increases only the rank of $w$, and only during the step-up in the mass of $w$.  For any $a \geq 0$, the amount of the increase is at most $(2/2^{a+b})(\lg(v.m) -\lg(w.m)) + a$, where variables take their values just before the step-up. 
\end{lemma}
\begin{proof}
The link changes the size, mass, and rank only of $v$ and $w$.  If the link is a left link, it decreases $w.m$ and does not change $w.s$, so the $b$-reduced rank of $w$ does not increase. The link increases $v.s$ and $v.m$ by $w.s$, so the $b$-reduced rank of $v$ does not increase.

Suppose the link is a right link.  The link increases $v.s$ but does not change $v.m$, so the $b$-reduced rank of $v$ does not increase.  The step-up in the mass of $w$ just before the link increases $w.r$ by 
\begin{align*}
\lg(\lg(v.m/w.s)/2^b)-\lg(\lg(w.m/w.s)/2^b) &\leq (2/2^{a+b})(\lg(v.m/w.s)-\lg(w.m/w.s)) + a\\
&= (2/2^{a+b})(\lg(v.m) - \lg(w.m)) + a
\end{align*}
by Lemma~\ref{L:inequality-2} and $v.m > w.m \geq w.s$.  After the link, the rank of $w$ is $\lg\lg((v.s+w.s)/w.s) \leq \lg\lg(v.m/w.s)$, so the link itself does not increase the rank of $w$.
\end{proof}

\begin{lemma}\label{L:pairing-pass}
In a multipass pairing heap, for any $a \geq 0$, the sum of the $b$-reduced rank increases caused by a pairing pass that does $k$ real right links is at most $(2/2^{a+b})\lg n + ak$.  
\end{lemma}

\begin{proof}
We sum the upper bound $(2/2^{a+b})(\lg(v.m) -\lg(w.m)) + a$ on the increase in $b$-reduced ranks given by Lemma~\ref{L:multipass-real-deletion-links} over the $k$ real right links $vw$ done during the pairing pass.  If $vw$ and $v'w'$ are two such links, with $v'$ right of $v$ on the root list, $w.m > v'.m$, where masses are evaluated just before any of the links.  Thus the sum of $\lg(v.m) -\lg(w.m)$ telescopes to at most $\lg n$, giving the lemma.
\end{proof}

\begin{lemma}
\label{L:multipass-delete-min-links}
In a multipass pairing heap, for any $c\geq 0$, the sum of the $b$-rank increases caused by a delete-min that does $k$ links is at most $(2(b+c+2)/2^b)\lg n + 4k/2^{b+c}$.
\end{lemma}

\begin{proof}
We sum the upper bound in Lemma~\ref{L:pairing-pass} over all pairing passes of the delete-min, but we choose $a$ as a function of the pairing pass.  By Lemma~\ref{L:few-pass-i-links}, at most $2k/2^i$ links are done in pairing pass $i$.

Setting $a_i = \max\{i-\lceil b+c \rceil, 0\}$ for pairing pass $i$ and summing the bound in Lemma~\ref{L:pairing-pass} over all pairing passes in the delete-min, we find that the sum of the $b$-rank increases is at most
\begin{align*}
&(2\lceil b+c \rceil/2^b)\lg n 
 + \sum_{i=1}^\infty\bigr((2/2^b)(1/2^i)\lg n + (2k/2^{\lceil b+c \rceil})(i/2^i)\bigr)\\
 &= (2(\lceil b+c\rceil+1)/2^b)\lg n + 4k/2^{\lceil b+c \rceil}\\
 &\leq (2(b+c+2)/2^b)\lg n + 4k/2^{b+c}\qedhere
\end{align*}
\end{proof}

Combining Lemmas~\ref{L:multipass-real-decrease-key-links},~\ref{L:multipass-real-insertion-links}, and~\ref{L:multipass-delete-min-links} and using the fact that twice the sum of the increases in $b$-reduced node ranks is an upper bound on the sum of the magnitudes of the change is $b$-reduced node ranks gives us Theorem~\ref{T:b-rank-variation}:

For any non-negative integers $b$ and $c$, in any sequence of multipass pairing heap operations starting with no heaps, the sum of the magnitudes of the changes in $b$-reduced node ranks is at most $4\lg((\lg n)/2^b)$ per decrease-key plus $(4(b+c+2)/2^b)\lg n + 2\lg((\lg n)/2^b) + 6/2^b$ per delete-min plus $8d/2^{b+c}$.

\section{Analysis of slim heaps}\label{S:slim-analysis}
\subsection{Basic properties of slim and smooth heaps}\label{S:slim-smooth-properties}

Now we turn to the analysis of slim and smooth heaps.  We start with slim heaps, and then adapt the analysis to smooth heaps.  We begin by giving two basic properties of both slim and smooth heaps.  The first is that during a delete-min each node wins at most one left link and one right link~\cite{HKST21,KS19}. For completeness we include a proof of this result.  

\begin{lemma}
\label{L:s-treap}~\cite{HKST21,KS19}
During a delete-min in a slim or smooth heap, each node wins at most one left link and at most one right link.
\end{lemma}

\begin{proof}
Suppose a node $v$ wins a right link with a node $w$.  After the link, either $v$ is rightmost on the root list or its right neighbor has equal or smaller key.  In the latter case, during the rest of the delete-min the right neighbor of $v$ can change, but the key of its right neighbor cannot increase.  Hence $v$ can only lose, not win, a link with its right neighbor during the rest of the delete-min.  The argument for left links is symmetric, with minor changes to account for the tie-breaking rule in leftmost locally maximum linking.
\end{proof}

By Lemma~\ref{L:s-treap}, the set of links done during a delete-min in a slim or smooth heap forms a binary tree that is heap-ordered by key and symmetrically ordered by position on the root list just after the original root is deleted.  We call this tree the \emph{treap} formed by the delete-min.  See Figure~\ref{F:treap-example}.

\begin{figure}[hb!]
\centering
\includegraphics[width=4in]{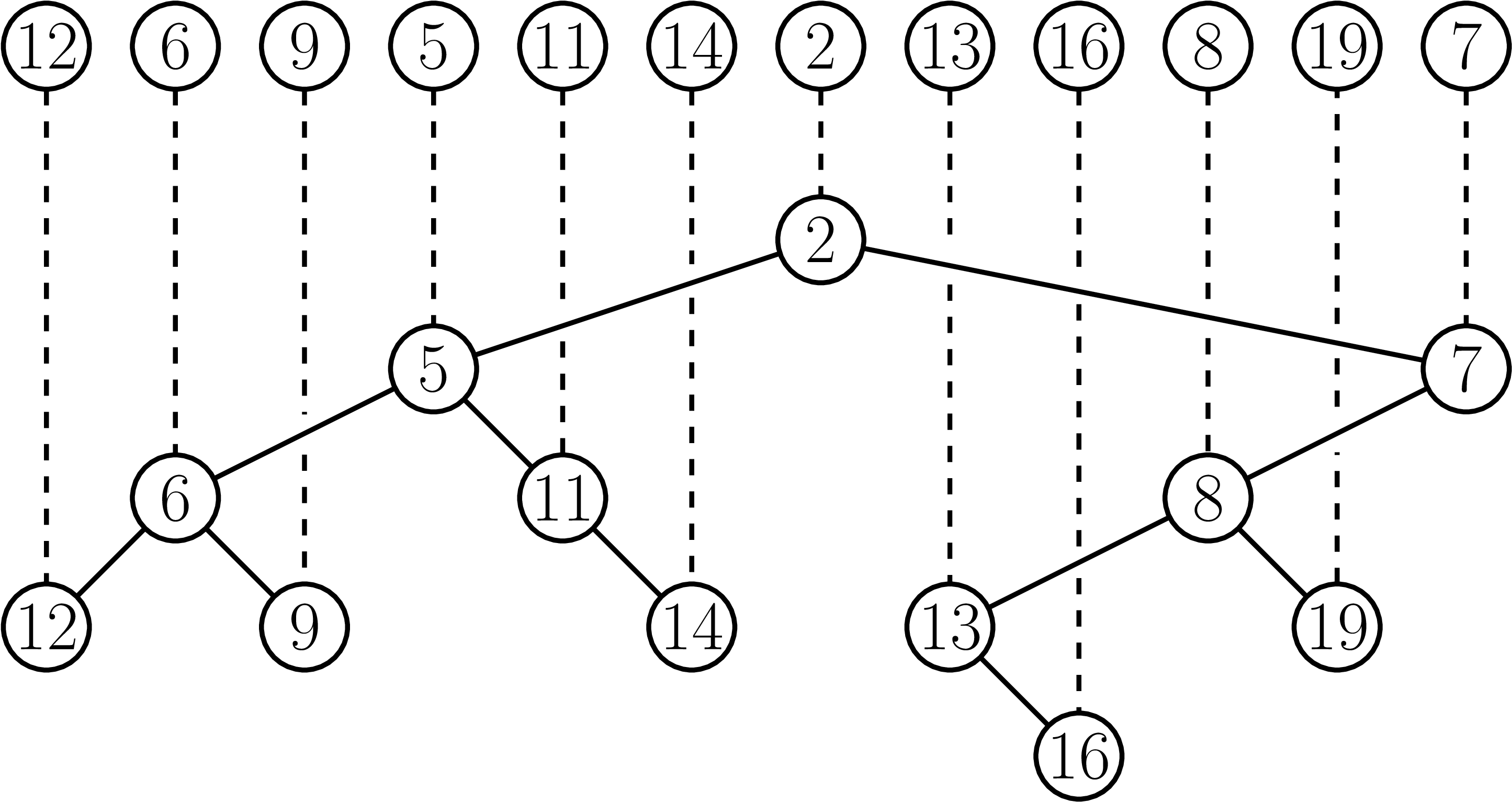}
% \Description{An illustration of the treap formed when the roots have keys (in order from left to right): 12, 6, 9, 5, 11, 14, 2, 13, 16, 8, 19, and 7.}
\caption{An example of a treap formed by leftmost locally maximum linking.  Original subtrees of the roots linked are not shown.  Symmetric order in the treap is left-to-right order on the root list before any linking.}
\label{F:treap-example}
\end{figure}

Our second result bounds the number of deletion links $d$ in terms of the number of real deletion links.  This is the first step in bounding the number of deletion links. 

\begin{lemma}\label{L:slim-smooth-link-bound}
In any sequence of operations on slim or smooth heaps starting with no heaps, $d$ is at most three per real deletion link plus three per insertion plus three per decrease-key. 
\end{lemma}
\begin{proof}
Consider a delete-min that does $k$ deletion links and in which $j$ of the $k+1$ new roots are temporary.  We consider two cases.  If $j < k/3$, at least $k -2j \geq k/3$ of the deletion links must be won by a permanent node, since by Lemma~\ref{L:s-treap} each of the $j$ temporary roots wins at most two of the $k$ links.  Each link won by a permanent node is either an f-link (never cut) or a k-link (cut by a decrease-key).  If it is an f-link, we charge $3$ to the loser of the link, which is a permanent node.  If it is a k-link, we charge $3$ to the decrease-key that cuts the link.  The total charge is at least $k$, the number of deletion links.

If $j \geq k/3$, at least $j - 1 \geq k/3 - 1$ deletion links must be lost by temporary nodes.  Each such link is either a real deletion link or a k-link.  In the former case we charge $3$ to the link; in the latter we charge $3$ to the decrease-key that cuts the link.  We also charge $1$ to the insertion of the node deleted by the delete-min, which is a temporary node.  The total charge is at least $k$.

Combining the cases gives the lemma.
\end{proof}

\subsection{The treap and its boundaries}\label{S:treap}

In slim and smooth heaps, the links done during a delete-min do not occur in separate passes, so the concept of $b$-reduced ranks is not helpful.  Instead of partitioning the links done during a delete-min into passes, we partition them in a way that depends on the treap $\cT$ whose edges are the links done during the delete-min.  In this section we study some properties of this treap and introduce some relevant terminology.

We define a \emph{boundary} of $\cT$ to be a partition of the set of nodes in $\cT$ into two nonempty \emph{sides}, such that every node on one side is smaller in symmetric order (with respect to $\cT$) than every node on the other side.  The former is the \emph{left side} of the boundary; the latter is the \emph{right side}.  Each node $x$ in $\cT$ except the largest in symmetric order defines a unique boundary whose left side contains $x$ and all smaller nodes in symmetric order and whose right side contains all nodes larger than $x$ in symmetric order.  We call this the \emph{boundary of $x$}. A link $vw$ done during the delete-min \emph{crosses} a boundary if exactly one of $v$ and $w$ is on each side of the boundary.  See Figure~\ref{F:alternating}.

\begin{figure}[h!]
\begin{center}
\includegraphics[width=3in]{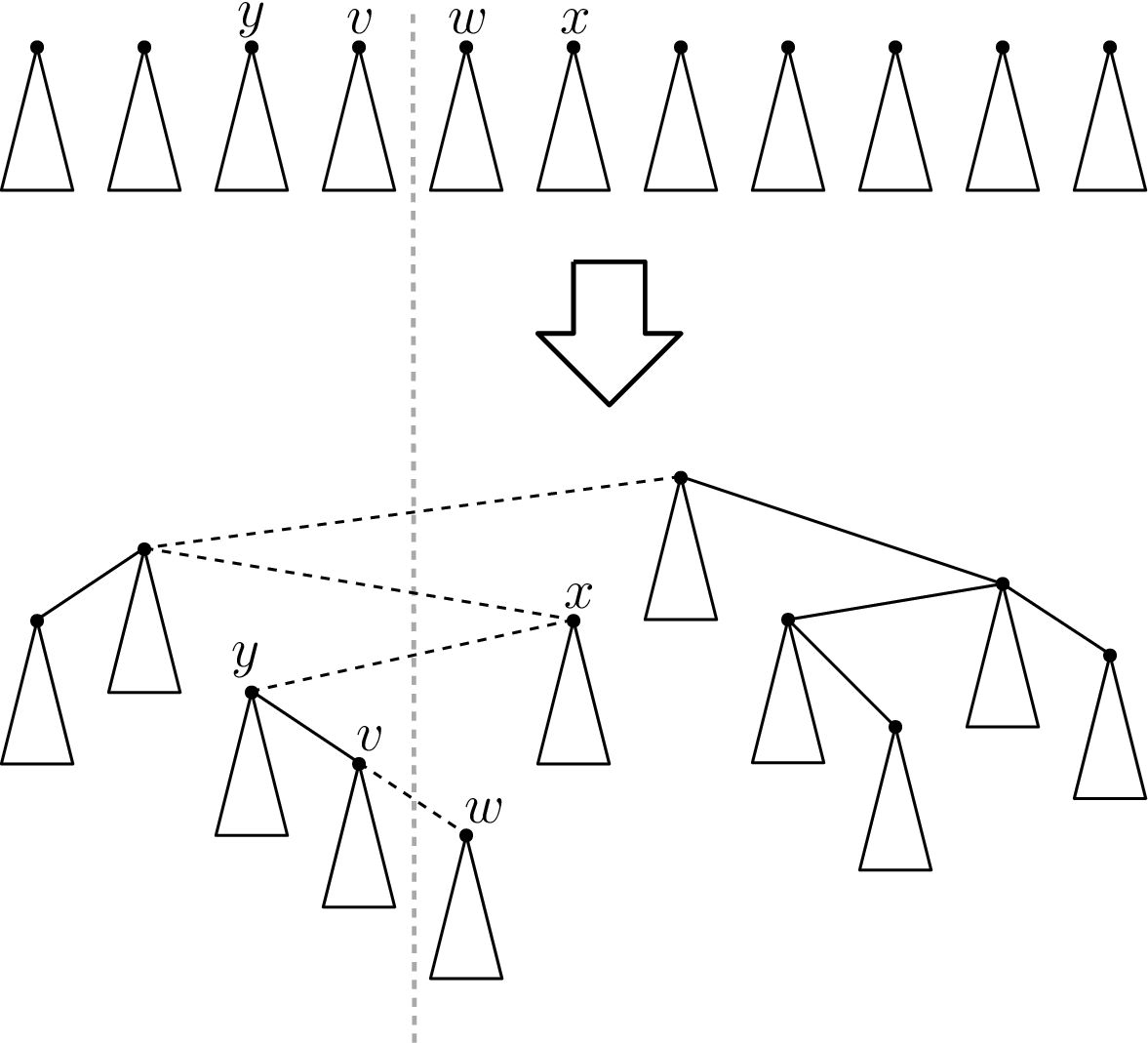}
\end{center}
% \Description{A example treap with emphasis on the links that cross a particular boundary. All the crossing links lie on a single path in the treap and alternate left and right links.}
\caption{Links crossing a boundary after a delete-min.  The dashed links cross the dashed grey boundary and alternate between left and right links. The boundary pictured is the boundary of $v$.}
\label{F:alternating}
\end{figure}

\begin{lemma}\label{L:crossing-links-path}
Let $\cT$ be the treap formed by a delete-min in a slim or smooth heap. Given the boundary of a node $x$ in $\cT$, all the links in $\cT$ that cross the boundary are on a single path in $\cT$ from the root of $\cT$ down to either $x$, if $x$ does not win a crossing link, or down to the smallest node greater than $x$ in symmetric order if $x$ does win a crossing link.  Along this path, the links crossing the boundary alternate between left and right.  If $x$ wins a right link in $\cT$, this link is the lowest link to cross the boundary.
\end{lemma}
\begin{proof}
Start at the root of $\cT$ and follow the path $P$ defined as follows: Let $v$ be the current node on the path, initially the root.  If $v$ is larger in symmetric order than $x$, replace $v$ by its left child in $\cT$; otherwise, replace $v$ by its right child in $\cT$.  Repeat until $v=x$.  If $x$ has a right child, replace $v$ by its right child, and then repeatedly replace $v$ by its left child until $v$ has no left child.

If $v$ is replaced by its right child, the link to its left child, as well as every link in the left subtree of $v$, does not cross the boundary; if $v$ is replaced by its left child, the link to its right child, as well as every link in the right subtree of $v$, does not cross the boundary.  If $v=x$, the link to its left child, and all links in the subtree of this child, do not cross the boundary; the link to the right child of $v$, if it exists, crosses the boundary, but no links in the subtree of this right child cross the boundary.  It follows by induction on the length of $P$ that every link crossing the boundary is on $P$.  If $vw$ is a link on $P$ that does not cross the boundary, both $v$ and $w$ are on the same side of the boundary.  It follows by induction on the length of $P$ that if $uw$ and $yz$ are successive links on $P$ that cross the boundary, both $w$ and $y$ are on the same side of the boundary, which means that exactly one of $uw$ and $yz$ is a left link. That is, the links crossing the boundary alternate between left and right.  If $x$ has no right child, $x$ is the bottom node on $P$.  If $x$ has a right child, say $y$, the bottom node on $P$ is the smallest node in symmetric order that is a descendant of $y$, which is the node in $\cT$ after $x$ in symmetric order.
\end{proof}

\subsection{Link order}
\label{S:link-order}

In our analysis of slim heaps, we shall make the simplifying assumption that if a temporary root $u$ wins both a left link $uv$ and a right link $uw$ during a delete-min, $u$ wins the right link first.  Suppose that in fact $u$ wins the left link first.  After both links, $w$ is the left neighbor of $v$ on the list of children of $u$, but if $u$ wins the right link first, the order is reversed: $w$ is the right neighbor of $v$.  In this case we restore the correct order when $u$ is deleted, by swapping the order of $v$ and $w$ on the root list before any links are done.  This guarantees that the set of links done in the actual algorithm is the same as the set of links we analyze.  In our analysis we shall need to account for the effect of such swaps.

\subsection{Ranks in slim heaps}\label{S:slim-ranks}

In slim heaps, the definitions of mass and rank depend on the future behavior of the algorithm, not just on its present and past behavior.  Specifically, masses and ranks during a delete-min depend on the treap $\cT$ formed by the delete-min, which is not completely determined until the last link of the delete-min.  We partition $\cT$ into maximal paths of real right links. Each node in $\cT$ is on exactly one such path.  Given such a path, if its top node (the one nearest the root of $\cT$) is the root itself or is the loser of a phantom link in $\cT$, we define all nodes on the path to be \emph{heavy}; otherwise all nodes on the path are \emph{light}.  Each node of $\cT$ retains its state (light or heavy) throughout the delete-min.  Each permanent node of $\cT$ forms a one-node path and is heavy, since any link it wins or loses is a phantom link.  Each real right link is between a pair of heavy nodes or a pair of light nodes.  Each loser of a real left link is light.  We define a real left link $vw$ to be \emph{anomalous} if $v$ is heavy and the subtree of $\cT$ rooted at $w$ contains at least one k-link.  See Figure~\ref{F:linking-example}.

\begin{figure}
\centering
\includegraphics[width=4.8in]{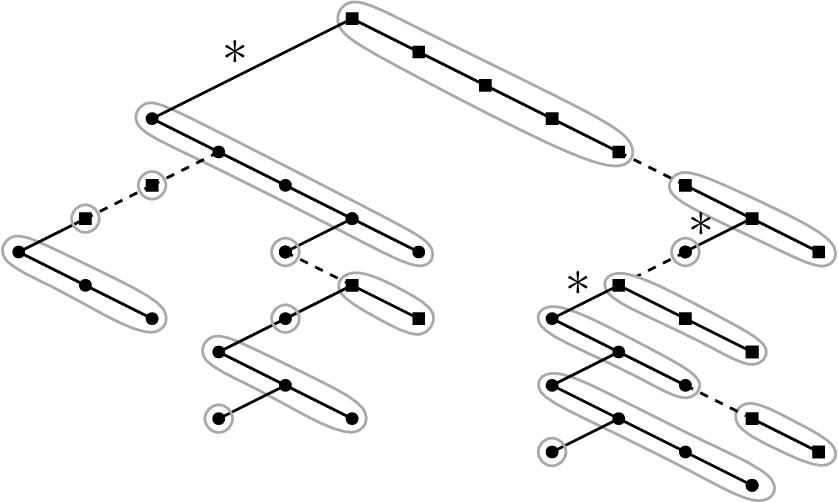}
% \Description{An example treap to illustrate the terms we have just defined: heavy nodes, light nodes, and anomalous links.}
\caption{An example of the treap formed by a delete-min.  Squares denote heavy nodes, circles light nodes.  Asterisks indicate anomalous links.  Solid links are real, dashed links phantom.  Each maximal path of real right links is circled.}
\label{F:linking-example}
\end{figure}

We define sizes in a slim heap as in a multipass pairing heap: The size $v.s$ of a temporary node $v$ is the number of temporary nodes that are descendants of $v$ and connected to $v$ by a path of real links.  During a delete-min, each root $v$, temporary or permanent, also has an \emph{extended size} $v.\bar{s}$.  We define extended sizes operationally.  Just after the original root is deleted in a delete-min, each temporary root has extended size equal to its size, and each permanent root has extended size $0$.  A link $vw$ increases the extended size of $v$ by that of $w$, whether the link is real or phantom.  (If the link is real, it also increases the size of $v$ by that of $w$.)   Equivalently, the extended size of a root $v$ during a delete-min is the number of temporary nodes that are descendants of $v$ and connected to it by a path of real links and links done during the current delete-min (before the present moment).  That is, phantom links done before the current delete-min are excluded.

We define the \emph{mass} $v.m$ of a temporary node $v$ as follows:  If $v$ is a  real child, its mass is one plus its size plus the sizes of its siblings to its right on its list of siblings.  Equivalently, its mass is the size of its parent $u$ just after $v$ lost the link $uv$.  If $v$ is a phantom child or a root not on the root list in the middle of a deletion, its mass is its size.  In these cases the definition of mass is the same as in multipass pairing heaps. 

If $v$ is a temporary root on the root list in the middle of a deletion, its mass depends on whether it is light or heavy.  If $v$ is heavy, its mass is its extended size plus the sum of the extended sizes of the roots to its right on the root list.  If $v$ is light, its mass is its size plus the sum of the extended sizes of the roots to its right on the root list.  There are two exceptions to this definition of mass:
\begin{itemize}
\setlength\itemsep{1em}
\item[(i)] Just before a temporary node $w$ loses a real right link $vw$, its mass becomes $v.m'$, the mass of $v$ just after the link.  After the link, the mass of $w$ becomes $v.s'$, the size of $v$ just after the link.
\item[(ii)] Just after a temporary node $v$ wins an anomalous link, its mass temporarily decreases to its size, then increases to its extended size plus the extended sizes of the nodes to its right on the root list. 
\end{itemize}

We call the increase in the mass of $w$ from $w.m$ to $v.m'$ in exception (i) and of that of $v$ from $v.s'$ to $v.m'$ in exception (ii) a \emph{step-up} in the mass of $w$ or $v$, respectively.     We call the decrease in the mass of $v$ in exception (ii) a \emph{step-down}.  We extend the terms ``step-up" and ``step-down" to the changes in rank resulting from the changes in mass.  Exception (i) is the same exception we used in the analysis of multipass pairing heaps.  The exceptions guarantee that all link ranks are non-negative.

A permanent node has no size, extended size, or mass, unless it is a root in the middle of a delete-min, when it has an extended size.  The extended size of a permanent node $v$ is zero at the beginning of a delete-min.  It becomes non-zero if (and only if) it has at least one descendant in $\cT$ that is a temporary node: If so, it will win a link during the delete-min that makes its extended size positive.

The definitions of node ranks, link ranks, $w$-shifts, $v$-shifts, and $u$-shifts are the same for slim heaps as for multipass pairing heaps. They are also stated in Section~\ref{S:terminology}.  In these definitions, if $vw$ is a real right link, $w.r_v$ is the rank of $w$ after its step-up just before the link, and if $vw$ is an anomalous link, $v.r_w'$ is the rank of $v$ after its step-down just after the link, but before its step-up.  

Now that we have the appropriate definitions, we can prove Lemma~\ref{L:multipass-link-ranks-nonnegative} for slim heaps.

\begin{lemma}\label{L:slim-link-ranks-nonnegative}
In a slim heap, every link rank is non-negative. 
\end{lemma}

\begin{proof}
Let $vw$ be a real deletion link with winner link $uv$.  The link rank of $vw$ is $v.r_u - u.r_v'$.  If $uv$ is an insertion, decrease-key, or anomalous link, the mass of $u$ just after link $uv$ is done is the size of $u$.  Thus $u.r_v' = 0$, so $v.r_u - u.r_v' \geq 0$. 

Suppose $uv$ is a right link or a non-anomalous left link.  Let unprimed and primed sizes and masses denote values just before and just after $uv$ is done, respectively.

If $uv$ is a right link, $v.r_u = \lg\lg(u.m'/v.s')$ due to the step-up in the mass of $v$ just before the link.  Also, $u.r_v' = \lg\lg(u.m'/u.s')$.  Since $v.s' = v.s$ and $u.s' = u.s + v.s$, $u.m'/v.s' \geq u.m'/u.s'$, giving  $v.r_u - u.r_v'\geq 0$.

If $uv$ is a left link, then $v$ is light.  Let $M$ be the sum of the extended sizes of all roots right of $u$ on the root list just before $uv$ is done.  If $u$ is light, then $v.m = v.s+u.\bar{s}+M \geq v.s+u.s+M = u.m'$. Since $v.s \leq u.s'$, $u.m'/u.s' \leq v.m/v.s$, giving $v.r_u - u.r_v' \geq 0$.

If $u$ is heavy, $v.\bar{s} = v.s$, since $uv$ is not anomalous.  This implies $v.m = v.s+u.\bar{s}+M = v.\bar{s}+u.\bar{s}+M=u.m'$, so in this case also $u.m'/u.s' \leq v.m/v.s$, giving $v.r_u - u.r_v' \geq 0$. 
\end{proof}

We proceed as in the analysis of multipass pairing heaps.

\subsection{Small-rank links in slim heaps}\label{S:small-rank-slim}

Our next step is to bound the number of real deletion links of small rank.  First we prove an analogue for slim heaps of Lemma~\ref{L:multipass-small-node-ranks} for multipass pairing heaps.

\begin{lemma}
\label{L:slim-small-node-ranks}
In a slim heap just after the root swaps in a delete-min (if any), the number of temporary roots on the root list with rank at most $1$ is at most $2.5\lg n + 1$.
\end{lemma}
\begin{proof}
Just after the swaps, there have been no links, so the size of each temporary root equals its extended size, and the extended size of each permanent root is zero.  The lemma follows by the proof of Lemma~\ref{L:multipass-small-node-ranks}.   
\end{proof}

Second we observe that Lemma~\ref{L:multipass-small-link-rank-decreases} holds for slim heaps:

\begin{lemma}
\label{L:slim-small-link-rank-decreases}
Let $vw$ be a real deletion link with winner link $uv$.  Suppose $vw$ has link rank less than $1$.  If both $u$ and $v$ have positive rank just after $uv$ is done, then $v.r' \leq v.r - 1$, where unprimed and primed variables take their values just before $uv$ is done (after the step-up in the mass of $v$ if there is one) and just after $uv$ is done, respectively. 
\end{lemma}
\begin{proof}
Since $u$ and $v$ have positive rank just after $uv$ is done, $uv$ is not an anomalous link.  This makes the proof of Lemma~\ref{L:multipass-small-link-rank-decreases} valid for slim heaps. \end{proof}

Third we use Lemma~\ref{L:slim-small-node-ranks} and Lemma~\ref{L:slim-small-link-rank-decreases} to obtain a bound on the number of real deletion links of rank less than $1$ in terms of the sum of the magnitudes of all node-rank decreases.

\begin{lemma}\label{L:slim-small}
In any sequence of slim heap operations starting with no heaps, the number of real deletion links of rank less than $1$ is at most six times the sum of the magnitudes of all node-rank decreases plus two per insertion plus two per decrease-key plus $15\lg n + 6$ per delete-min. 
\end{lemma}
\begin{proof}
Let $vw$ with winner link $uv$ be a real deletion link with link rank less than $1$.  Since $v$ wins at most two links during a deletion, a given link $uv$ is the winner link of at most two such links $vw$.  We shall charge $vw$ to some event related to $uv$, doubling our estimate to account for the possibility that $uv$ is the winner link of two different links.    

If $uv$ is an insertion or decrease-key link, we charge $vw$ to $uv$.  After doubling, the total of such charges is at most two per insertion and two per decrease-key.  If $uv$ is a deletion link and $u$ or $v$ has rank less than $1$ just after the root swaps in the delete-min that does $uv$, we charge $vw$ to one of $u$ and $v$ that has rank less than $1$ just after the swaps, choosing either if both have rank less than $1$.  Since a given node can win at most two links and lose at most one link during a delete-min, the total charge per node for the winner links done during a given delete-min is at most six, including the doubling.  By Lemma~\ref{L:slim-small-node-ranks}, the total of such charges is at most $15\lg n + 6$ per delete-min.  

The remaining possibility is that $uv$ is a deletion link and both $u$ and $v$ have rank at least $1$ just after the swaps in the delete-min that does $uv$.  We claim that in this case the sum of the magnitudes of the rank decreases of at least one of $u$ or $v$ from just after the root swaps in the delete-min that does $uv$ until just after $uv$ is done must be at least $1$.  If this is not already true just before $uv$ is done, then both $u$ and $v$ have positive rank just before $uv$ is done, but then by Lemma~\ref{L:slim-small-link-rank-decreases} doing $uv$ decreases the rank of $v$ by at least $1$, making it true.  We charge $vw$ to the sum of the magnitudes of the node-rank decreases of one of $u$ and $v$ whose sum of node-rank decreases during the delete-min that does $uv$ is at least $1$.  The total of such charges is at most six per unit of node-rank decrease.  

Adding these bounds gives the lemma.
\end{proof}

To finish the task of bounding the number of real deletion links of small rank, we use an  analogue for slim heaps of Theorem~\ref{T:multipass-node-rank-changes}.  The bound for slim heaps is simpler in that there is no additive term in $d$, the total number of deletion links.  We state the bound here and prove it in Section~\ref{S:node-rank-changes-slim}.

\begin{theorem}\label{T:slim-node-rank-changes}
In any sequence of slim heap operations starting with no heaps, the sum of the magnitudes of all node-rank changes is $\OO(\lg\lg n)$ per decrease-key plus $\OO(\lg n)$ per delete-min.  
\end{theorem}

Since any link is the winner link of at most two real deletion links, Lemma~\ref{L:slim-small} and Theorem~\ref{T:slim-node-rank-changes} combine to give us our desired bound on the number of real deletion links of small rank:

\begin{theorem}\label{T:slim-small}
In any sequence of slim heap operations starting with no heaps, the total number of real deletion links with rank less than $1$ is $\OO(1)$ per insertion plus $\OO(\lg\lg n)$ per decrease-key plus $\OO(\lg n)$ per delete-min. 
\end{theorem}

\subsection{Large-rank links in slim heaps}\label{S:large-rank-slim}

Having bounded the number of real deletion links of small rank, we now bound the number of such links of large rank.  As in the case of multipass pairing heaps, we do this by bounding the sum of link ranks using the link-rank equality, which states that if $vw$ is a real deletion link with winner link $uv$, then
\begin{align*}
v.r_u-u.r_v' = (v.r_w' - u.r_w') + (v.r_u - v.r_w') + (u.r_w' - u.r_v')
\end{align*}

That is, the rank of $vw$ is the sum of its $w$-shift, its $v$-shift, and its $u$-shift.  Summing over all real deletion links, the sum of the ranks of the real deletion links equals the sum of their $w$-shifts, $v$-shifts, and $u$-shifts.  Lemma~\ref{L:w-shifts} is a bound on the sum of the $w$ shifts.  It remains to bound the sum of the $v$-shifts and the sum of the $u$-shifts.  Bounding the sum of the $v$-shifts is easy, because during a single delete-min a node can win only two links.  Bounding the sum of the $u$-shifts uses the treap structure of the links done by a delete-min.  We rely on Theorem~\ref{T:slim-node-rank-changes}, which we stated in the previous section and will prove in the next one.

\begin{lemma}\label{L:slim-v-shifts}
In any sequence of slim heap operations starting with no heaps, the sum of the $v$-shifts of the real deletion links is $\OO(\lg\lg n)$ per decrease-key plus $\OO(\lg n)$ per delete-min. 
\end{lemma}

\begin{proof}
Let $vw$ be a real deletion link with winner link $uv$.  The $v$-shift of $vw$ is $v.r_u - v.r_w'$, the net decrease in the rank of $v$ from the time $t_1$ just before link $uv$ is done (after the step-up in the mass of $v$ if it is a right link) until the time $t_2$ just after $vw$ is done.  We call $[t_1, t_2]$ the $v$-shift time interval of $vw$.  At most two links are won by $v$ during a single delete-min.  If $v$ wins two such links, arbitrarily label them $1$ and $2$.  For different real deletion links $vw$ with the same label won by $v$, the corresponding $v$-shift time intervals are disjoint.  It follows that for a given $v$ the sum of the $v$-shifts of all real deletion links $vw$ labeled $1$ is at most the sum of the decreases in $v.r$ over the sequence of heap operations, as is the sum of the $v$-shifts of all real deletion links $vw$ labeled $2$.  Summing over all $v$ and applying Theorem~\ref{T:slim-node-rank-changes} gives the lemma.     
\end{proof}

The $u$-shifts are far more challenging to bound, because an unbounded number of the real links done when node $u$ is deleted can have overlapping $u$-shift time intervals.  For a given $u$, we bound the sum of the $u$-shifts using the cancellation among these shifts that occurs because of the treap structure of the links during the deletion of $u$.

\begin{lemma}
\label{L:slim-u-shifts}
In any sequence of slim heap operations starting with no heaps, the sum of the $u$-shifts of the real deletion links is $\OO(\lg\lg n)$ per decrease-key plus $\OO(\lg n)$ per delete-min. 
\end{lemma}

\begin{proof}
Consider a delete-min that deletes $u$.  Let $\cT$ be the treap formed by the links done during the delete-min.   Let $v_1, v_2,\dots, v_k$ be the temporary nodes in $\cT$, in symmetric order in $\cT$.  This is their order on the root list just after the swaps done to correct for the linking order: See Section~\ref{S:link-order}.  For $1 \leq i \leq k$, let $u.r'_i$ be the value of $u.r$ just after $u$ won the link $uv_i$ cut by the delete-min.  For $1 \leq i < k$, let  $\Delta_i = u.r'_i - u.r'_{i+1}$, and let $I_i$ be the interval of time between just after $uv_{i+1}$ is done and just after $uv_i$ is done.  The value of $|\Delta_i|$ is the magnitude of the change in $u.r$ from the beginning of $I_i$ to the end of $I_i$, and is at most the sum of the magnitudes of the changes in $u.r$ during $I_i$.

Before the swaps, the $v_i$ are in order on the root list by the time they were linked to $u$, latest to earliest from left to right.  This is not true after the swaps, but the swaps do not change the order too much, as we now show.  Any swaps are of node-disjoint pairs of adjacent nodes on the root list. A swap that involves a permanent root has no effect on the order of the temporary nodes, so we can ignore all such swaps.  Consider $v_i$ and $v_{i+1}$.  At most two temporary nodes, $v_{i-1}$ and $v_{i+2}$, are in between $v_i$ and $v_{i+1}$ before the swaps. If both $v_{i-1}$ and $v_{i+2}$ are in between $v_i$ and $v_{i+1}$, then $|\Delta_i| \leq |u.r'_i - u.r'_{i-1}| + |u.r'_{i-1} - u.r'_{i+2}| + |u.r'_{i+2}-u.r'_{i+1}|$.  If $v_{i-1}$ but not $v_{i+2}$ is in between, then $|\Delta_i| \leq |u.r'_i - u.r'_{i-1}| + |u.r'_{i-1} - u.r'_{i+1}|$.  Similarly, if $v_{i+2}$ but not $v_{i-1}$ is in between, then $|\Delta_i| \leq |u.r'_i - u.r'_{i+2}| + |u.r'_{i+2}-u.r'_{i+1}|$.  Finally, if neither is in between, $|\Delta_i| \leq |u.r'_i - u.r'_{i+1}|$.  In each case $|\Delta_i|$ is at most the sum of at most three terms, each of which is the change or the negative of the change in $u.r$ from the time $u$ won a link with one of $v_{i-1}$, $v_i$, $v_{i+1}$, or $v_{i+2}$ until it won its next link with one of $v_{i-1}$, $v_i$, $v_{i+1}$, or $v_{i+2}$.  This implies that $\sum_{i=1}^{k-1}|\Delta_ i|$ is at most three times the sum of the magnitudes of all the changes in $u.r$ from the time $u$ was inserted until it was deleted.       

Let $v_sv_t$ be a real link done during the delete-min that deletes $u$.  If $t < s$, the link is a left link, and its $u$-shift is $\sum_{i=t}^{s-1} \Delta_i$.  For $t \leq i < s$, we call $\Delta_i$ the \emph{$i$-term} of $v_sv_t$.  If $t>s$, the link is a right link, and its $u$-shift is $\sum_{i=s}^{t-1} -\Delta_i$.  For $s \leq i < t$, we call $-\Delta_i$ the \emph{$i$-term} of $v_sv_t$.
We represent each $u$-shift as the sum of its $i$-terms and add them up.  We add these terms in groups, to maximize the cancellation of the plus and minus terms.

Recall from Section~\ref{S:treap} that a \emph{boundary} of $\cT$ is a partition of the nodes of $\cT$ into two nonempty parts, with all nodes in the left part less in symmetric order than all nodes in the right part.  The boundary of $v_i$ is the boundary whose largest left-part node in symmetric order is $v_i$.  Consider a real link in $\cT$ between $v_s$ and $v_t$ with $s < t$.  This link crosses the boundary of $v_i$ if and only if $s \leq i < t$.  If it does cross this boundary, its $i$-term contributes $\Delta_i$ or $-\Delta_i$ to the sum of the $u$-shifts if it is a left or right link, respectively.

By Lemma~\ref{L:crossing-links-path}, all the links crossing the boundary of $v_i$ are on a single path in $\cT$ from the root to $v_i$, or to its right child if it has a right child, and the links on this path that cross the boundary alternate between left and right links.  Given a real link $v_sv_t$ that crosses the boundary of $v_i$, we charge the $i$-term of $v_sv_t$ to $v_i$ if the path from $v_s$ to $v_i$ in $\cT$ contains only real links; if not, we charge the $i$-term of $v_sv_t$ to the topmost (closest to $v_s$) temporary node that loses a phantom link on the path from $v_s$ to $v_i$ in $\cT$.
We call the $i$-terms charged in the first way \emph{type-1 charges} and those charged in the second way \emph{type-2 charges}.

Consider the type-1 charges to a node $v_i$.  Each is of a link on a path of real links in $\cT$ whose bottommost node is $v_i$ or its right child.  Along this path, links crossing the boundary of $v_i$ alternate left and right by Lemma~\ref{L:crossing-links-path}.  The number of crossing left links is thus within one of the number of crossing right links.  This implies that the sum of the type-1 charges to $v_i$ is $-\Delta_i$, $0$, or $\Delta_ i$.  Summing over all $v_i$, the sum of the type-1 charges is at most $\sum_{i=1}^{k-1} |\Delta_i|$, which is at most three times the sum of the magnitudes of the changes in $u.r$ by the argument bounding the effect of the root swaps.  Summing over all $u$, the sum is within the bound of the lemma by Theorem~\ref{T:slim-node-rank-changes}(iii).

It remains to bound the type-2 charges.  Each of these is to a temporary node $v_j$ that lost a k-link in $\cT$.  We shall charge this k-link for all the type-2 charges to $v_j$.

Let $v_a$ and $v_b$ be the descendants of $v_j$ in $\cT$ with $a$ smallest and $b$ largest.  If $v_sv_t$ is a link with an $i$-term that is a type-2 charge to $v_j$, $v_t$ is a proper ancestor of $v_j$ in $\cT$.  Furthermore the $i$-terms of $v_sv_t$ charged to $v_j$, all of type 2, are exactly those such that $a \leq i < b$.  Finally, all such $i$-terms count positively if the link is a left link or negatively if the link is a right link.  Thus the sum of these charges for $v_sv_t$ is $u.r'_a-u.r'_b$, $0$, or $u.r'_b-u.r'_a$.  The set of links with an $i$-term that is a type-2 charge to $v_j$ are on a single path of real links in $\cT$, along which the links crossing the boundary of $v_j$ alternate left and right.  Hence the number of such links that are left is within one of the number that are right, and the sum of the type-2 charges to $v_j$ is $u.r'_a-u.r'_b$, $0$, or $u.r'_b-u.r'_a$. The magnitude of this sum is at most $\max\{u.r'_a, u.r'_b\}$.  We charge this amount to the k-link lost by $v_j$ in the delete-min that deletes $u$.  By the argument in the proof of Lemma~\ref{L:k-link-shift}, the sum of the type-2 charges over all delete-mins is $\OO(\lg\lg n)$ per decrease-key plus $3$ per delete-min.  Multiplying by a factor of three to account for the effect of root swaps, the sum is within the bound of the lemma.      
\end{proof}

Now we can prove the analogue of Theorem~\ref{T:multipass-link-rank-sum} for slim heaps: 

\begin{theorem}
\label{T:slim-link-rank-sum}
In any sequence of multipass pairing heap operations starting with no heaps, the sum of the link ranks of the real deletion links is $\OO(\lg\lg n)$ per decrease-key plus $\OO(\lg n)$ per delete-min.
\end{theorem}

\begin{proof}
We sum the link-rank equality over all real deletion links and apply Lemmas~\ref{L:w-shifts},~\ref{L:slim-v-shifts}, and~\ref{L:slim-u-shifts}.
\end{proof}

Since all link ranks are non-negative by Lemma~\ref{L:slim-link-ranks-nonnegative}, the bound in Theorem~\ref{T:slim-link-rank-sum} is also a bound on the number of real deletion links with rank at least $1$. Adding this bound to the bound in Theorem~\ref{T:slim-small} on the number of such links that have rank less than $1$ and applying Lemma~\ref{L:slim-smooth-link-bound} gives us our desired bound on the total number of deletion links:

\begin{theorem}\label{T:slim-deletion-link-bound}
In any sequence of slim heap operations starting with no heaps, the total number of deletion links is $\OO(1)$ per insertion plus $\OO(\lg\lg n)$ per decrease-key plus $\OO(\lg n)$ per delete-min. 
\end{theorem}

Theorem~\ref{T:slim-deletion-link-bound} gives our bound on the efficiency of slim heaps:

\begin{theorem}\label{T:slim-time-bound}
Any sequence of slim heap operations starting with no heaps takes $\OO(\lg\lg n)$ amortized time per decrease-key, $\OO(\lg n)$ amortized time per deletion, and $\OO(1)$ amortized (and worst-case) time for each other heap operation. 
\end{theorem}

\subsection{Node-rank changes in slim heaps}\label{S:node-rank-changes-slim}

Our remaining task is to prove Theorem~\ref{T:slim-node-rank-changes}.  We proceed as in the proof of Theorem~\ref{T:multipass-node-rank-changes}, with the addition of a lemma that bounds the effect of root swaps at the beginning of a delete-min.

The sum of node ranks is initially $0$.  We need to bound increases in this sum. Creation of a node does not increase the sum of node ranks, nor does a cut, nor does deletion of the root at the beginning of a delete-min.  The next lemma bounds the increase caused by swapping node-disjoint pairs of adjacent roots on the root list during a delete-min, to correct for changes in linking order as discussed in Section~\ref{S:link-order}.

\begin{lemma}\label{L:swapping}
In a slim heap, swapping node-disjoint pairs of adjacent roots on the root list just after deleting the original root in a delete-min increases the sum of node ranks by at most $2\lg n$.
\end{lemma}
\begin{proof}
Let variables take their values before any swaps.  Let $v$ and $w$ be two adjacent roots to be swapped, with $v$ left of $w$.  Since no links have been done, the extended size of each temporary root is its size, and the extended size of each permanent root is $0$.  Hence if either $v$ or $w$ is permanent, the swap has no effect on ranks.  If $v$ and $w$ are both temporary, the swap of $v$ and $w$ increases the mass of $w$ to that of $v$ and decreases the mass of $v$.  We ignore the decrease, which does not increase the rank of $v$.  The increase in the rank of $w$ caused by the swap is $\lg\lg(v.m/w.s)-\lg\lg(w.m/w.s) \leq 2\lg(v.m/w.s) - 2\lg(w.m/w.s)= 2\lg v.m - 2\lg w.m$.  If $v'$, $w'$ is another pair of temporary nodes to be swapped that are right of pair $v$ and $w$, then $\lg w.m > \lg v'.m$.  If follows that the sum of the increases in ranks caused by all the swaps is bounded above by a telescoping sum whose total is at most $2\lg n$.  
\end{proof}

Phantom decrease-key links and phantom insertion links change no node ranks. The next two lemmas are the analogues for slim heaps of Lemma~\ref{L:multipass-real-decrease-key-links} and Lemma~\ref{L:multipass-real-insertion-links}. 

\begin{lemma}\label{L:slim-real-decrease-key-links}
In a slim heap, a real decrease-key link $vw$ increases the node rank only of $w$, by at most $\lg\lg n$. 
\end{lemma}
\begin{proof} The proof of Lemma~\ref{L:multipass-real-decrease-key-links} with $b=0$ is valid for slim heaps as well.  
\end{proof}

\begin{lemma}\label{L:slim-real-insertion-links}
In any sequence of slim heap operations starting with empty heaps, the sum of node rank increases caused by real insertion links is at most $\lg\lg n$ per decrease-key plus at most $\lg\lg n + 3$ per delete-min.
\end{lemma}
\begin{proof}
The proof of Lemma~\ref{L:multipass-real-insertion-links} with $b=0$ is valid for slim heaps as well.
\end{proof}

Now we come to the biggest challenge, bounding the increase in node ranks caused by deletion links and the associated step-ups.  Because of the introduction of extended sizes, phantom links as well as real links can cause such increases.

\begin{lemma}\label{L:slim-phantom-link-winners}
In any sequence of slim heap operations starting with empty heaps, the sum of increases in node ranks caused by phantom deletion links is at most $\lg\lg n$ per decrease-key plus at most $3$ per delete-min.
\end{lemma}

\begin{proof}
Let $vw$ be a phantom deletion link.  The link increases the extended size of $v$ by that of $w$, which leaves the masses of all roots other than $v$ and $w$ unchanged, and leaves their ranks unchanged as well.  If  $w$ is temporary, it has  rank $0$ after the link, so the link does not increase the rank of $w$.  Suppose $v$ is temporary.  The link does not change the size of $v$.  If $vw$ is a right link, it does not increase the mass of $v$, so it does not increase the rank of $v$.

If $vw$ is a left link, it can increase the mass of $v$ only if $v$ is heavy, and then only if $w$ has positive extended size just before the link.  For $w$ to have positive extended size, it must have a temporary descendant in the treap $\cT$ formed by the delete-min that does $vw$.  Let $y$ be the nearest temporary descendant of $w$ in $\cT$.  The link $xy$ in $\cT$ lost by $y$ is a k-link: If $y=w$, $xy=vw$ is a k-link because it is phantom and both $v$ and $w$ are temporary; if $y \neq w$, $xy$ is a k-link because $x$ is permanent and $y$ is temporary.  The increase in the rank of $v$ caused by $vw$ is at most $\lg\lg n$, where $n$ is the size of the heap containing $v$ and $w$ when $vw$ is done.  Link $vw$ is uniquely determined by $xy$, since $w$ is the nearest temporary ancestor of $y$ in $\cT$.  We charge the increase in the rank of $v$ to $xy$ and apply Lemma~\ref{L:k-link-shift}.
\end{proof}

\begin{lemma}\label{L:slim-real-link-winners}
In a slim heap, a real deletion link $vw$ increases no node rank (although its associated step-up might).
\end{lemma}
\begin{proof}
Let $vw$ be a real deletion link.  The link increases the extended size of $v$ by that of $w$, which leaves the masses of all roots other than $v$ and $w$ unchanged, and leaves their ranks unchanged as well.

If $vw$ is a real right link, $v$ and $w$ are both heavy or both light.  In either case the link increases the size of $v$ but not its mass, which cannot increase its rank.  The step-up just before the link can increase the mass and hence the rank of $w$, but we consider step-ups separately.  The link itself does not increase the rank of $w$, since the mass of $w$ after the link is at most the mass of $v$ after the link, and the link does not change the size of $w$.

Suppose $vw$ is a real left link.  Then $w$ is light.  Whether $vw$ is anomalous or not, the link does not change the size of $w$, and it decreases the mass of $w$ to the size of $v$, resulting in no increase in the rank of $w$.  Let variables take their values just before the link.  If $v$ is light, the link increases both the size and mass of $v$ by $w.s$.  This does not increase the rank of $v$.  If $v$ is heavy but $vw$ is not anomalous, $w.s=w.\bar{s}$, so again the link increases the size and mass of $v$ by $w.s$, resulting in no increase in the rank of $v$.  Finally, if $vw$ is anomalous, the rank of $v$ just after the link is $0$ by exception (ii).
\end{proof}

It remains to bound the effect of step-ups.

\begin{lemma}\label{L:slim-anomalous-step-ups}
In any sequence of slim heap operations starting with empty heaps, the sum of increases in node ranks caused by step-ups during anomalous links is at most $\lg\lg n$ per decrease-key plus at most $3$ per delete-min.
\end{lemma}

\begin{proof}
The proof is just like the proof of Lemma~\ref{L:slim-phantom-link-winners}.
Let $vw$ be an anomalous link.  We charge the step-up in the rank of $v$ to a k-link done during the delete-min that does $vw$.  Let $\cT$ be the treap formed by this delete-min.  Since $vw$ is anomalous, there is a k-link in the subtree of $\cT$ rooted at $w$.  Choose such a link $xy$ such that there is no other such link on the path in $\cT$ from $x$ to $w$.  The path in $\cT$ from $x$ to $w$ cannot contain an anomalous link, since the winner of such a link would be on a path of real right links of heavy nodes whose top node lost a k-link, which would also be on the path, contradicting the choice of $xy$.  It follows that $vw$ is the anomalous link nearest $x$ on the path in $\cT$ from $x$ to the root.  Hence $vw$ is uniquely determined by $xy$.  We charge the step-up in the rank of $v$ to $xy$ and apply Lemma~\ref{L:k-link-shift}.
\end{proof}

\begin{lemma}\label{L:slim-heavy-step-ups}
During a single delete-min in a slim heap, the sum of increases in node ranks caused by step-ups before real right links lost by heavy nodes is at most $2\lg n$.
\end{lemma}
\begin{proof}
Let $vw$ be a real right link between heavy nodes during a delete-min.  By the assumption on linking order (Section~\ref{S:link-order}), $v$ has not won or lost a link since it became a root at the beginning of the delete-min, so its mass and size are unchanged since the root swaps at the beginning of the delete-min (if any).  Let $u$ be the nearest temporary root right of $v$ just after the swaps.  Node $u$ is $w$ or left of $w$ just after the swaps.  Node $w$ is heavy, so its mass just before the link $vw$ is the mass of $u$ just after the swaps, since all temporary nodes between $u$ and $w$ became descendants of $w$ before link $vw$. The step-up in the mass of $w$ thus increases its mass from $u.m$ to $v.m$, where these variables take their values just after the swaps.  This increases the rank of $w$ by $\lg\lg(v.m/w.s)- \lg\lg(u.m/w.s) \leq 2\lg(v.m)-2\lg(u.m)$, where $w.s$ is the size of $w$ just before the link.  Summing this bound over all real right links between heavy nodes done during the delete-min, the sum is bounded above by a telescoping sum that totals at most $2\lg n$.   
\end{proof}

We shall obtain the same bound as in Lemma~\ref{L:slim-heavy-step-ups} for increases in node ranks caused by step-ups before right links between light nodes, but we need a more complicated argument.  Let $P$ be a maximal path of real right links of light nodes in $\cT$.  Let $v$ be the topmost node on this path.  Since $v$ is light, it is a real child in $\cT$.  Let $uv$ be its link to its parent.  This link is a left link, since otherwise $P$ would not be maximal.  We shall charge the rank increases of the nodes on $P$ to the future delete-min that deletes $u$.

\begin{lemma}
\label{L:slim-light-right-links-some}
Let $P$ be a maximal path of real right links of light nodes in the treap $\cT$ formed by some delete-min in a slim heap.  Let $v$ be the top node on $P$, and let $uv$ be the real left link lost by $v$ during the delete-min.  Then the sum of the increases in node ranks caused by step-ups before links on $P$ is at most $2\lg u.s'-2\lg u.s$, where $u.s$ and $u.s'$ are the sizes of $u$ just before and just after $uv$, respectively.
\end{lemma}
\begin{proof}
Let $xy$ be a link on $P$.  Just before the link, $u$ is the next root right of $y$ on the root list.  Let $x.s$ and $y.s$ be the sizes of $x$ and $y$ just before the link, respectively, and let $M = y.m - y.s$. That is, $M$ is the sum of the extended sizes of $u$ and all the roots to its right just before the link.  The step-up before the link increases the rank of $y$ by $\lg\lg((x.s+y.s+M)/y.s)-\lg\lg((y.s+M)/y.s) \leq 2\lg((x.s+y.s+M)/(y.s+M))$.  Since $u.s \leq M$, the step-up in the rank of $y$ increases it by at most $2\lg((x.s+y.s+u.s)/(y.s+u.s)) = 2\lg(x.s+y.s+u.s)-2\lg(y.s+u.s)$. This bound telescopes when summed over the links on $P$. Let $w$ be the bottom node on $P$.  Node sizes do not decrease.  It follows that the sum of the bounds on the increases in node ranks over all links on $P$ is at most $2\lg(v.s'+u.s)-2\lg(w.s+u.s) = 2\lg u.s'-2\lg(w.s+u.s) \leq 2\lg u.s'-2\lg u.s$ , where $v.s'$ is the size of $v$ just after link $uv$.
\end{proof}

\begin{lemma}\label{L:slim-light-right-links-all}
In slim heaps, the sum of node rank increases caused by real right links of light nodes is at most $2\lg n$ per delete-min.
\end{lemma}

\begin{proof}
We sum the bound in Lemma~\ref{L:slim-light-right-links-some} over all real left links cut by a single delete-min.  Consider a delete-min that deletes node $u$.  Let $v_0, v_1,\dots, v_k$ be the children of $u$ in left-to-right order whose real left links $uv_i$ are cut by the deletion of $u$.  The sum of the bound in Lemma~\ref{L:slim-light-right-links-some} over the links $uv_i$ is the sum of $2\lg u.s_i' - 2\lg u.s_i$, where $u.s_i$ and $u.s_i'$ are the sizes of $u$ just before and just after $uv_i$, respectively.  Since $u.s_i \geq u.s_{i+1}'$, the sum is bounded above by a telescoping sum that sums to at most $2\lg n$.
\end{proof}

The results of this section combine to prove Theorem~\ref{T:slim-node-rank-changes}: In any sequence of slim heap operations starting with no heaps, the sum of the magnitudes of all node-rank changes is $\OO(\lg\lg n)$ per decrease-key plus $\OO(\lg n)$ per delete-min.

\section{Analysis of smooth heaps}
\label{S:smooth-analysis}

In this section we adapt the analysis of Section~\ref{S:slim-analysis} to smooth heaps.  The results of Sections~\ref{S:slim-smooth-properties} and~\ref{S:treap} hold for both slim heaps and smooth heaps.  Recall that the only difference between slim and smooth heaps is the position of the loser of a link on the list of children of a winner: In a slim heap, the loser becomes the new leftmost child of the winner, but in a smooth heap the winner of a deletion link becomes the new leftmost or rightmost child of the winner if the link is a left or a right link, respectively.  In both slim and smooth heaps the winner of an insertion or decrease-key link becomes the new leftmost child of the winner.  We call a child a \emph{left} or \emph{right child} if when it is linked to its parent it becomes the new leftmost or rightmost child of its parent, respectively.  The left children of a node are in decreasing order (latest to earliest) by link time, left to right on their list of siblings.  The right children are in increasing order by link time.  We say a left (right, respectively) child is \emph{interior} to all its leftward (rightward) siblings.

When a root is deleted in a delete-min, its list of children becomes a list of roots.  We call such a root a \emph{left} or \emph{right} root if it was a left or right child of its parent, respectively.  A left (right, respectively) root is \emph{interior} to all roots to its left (right) in the list of roots.  In the treap $\cT$ formed by a delete-min, the \emph{central boundary} is the boundary whose left side contains the left roots and whose right side contains the right roots.  The links crossing the central boundary are exactly those between a left root and a right root.  We call these links \emph{crossing links}.  We call a link between two left-side or two right-side roots a \emph{left-side} or \emph{right-side} link, respectively.  Let $x$ be the largest left-side root in symmetric order in $\cT$ and let $y$ be the smallest right-side root in symmetric order in $\cT$.  By Lemma~\ref{L:crossing-links-path}, the crossing links are on a single path in the treap from the root to $x$, if $x$ does not win a crossing link in $\cT$, or to $y$, if $x$ does win a crossing link in $\cT$.  We call this path the \emph{central path}. If there are no crossing links at all, the central path contains only the root of $\cT$.  In this case either all the links in $\cT$ are left-side links, or they are all right-side links.
We call a link \emph{central} if it is on the central path, \emph{non-central} otherwise.  Each central link won by a left root is a right link; each central link won by a right root is a left link.

In our analysis, we shall assume in a delete-min that all the non-central links are done before all the central links, and that every left (right, respectively) root that wins two non-central links wins the right (left) link first. These assumptions are without loss of generality, since in a smooth heap any bottom-up linking order during a delete-min that produces the same treap produces the same order of children of each node.  (This is \emph{not} true in a slim heap, since if a node wins both a left and a right link the order of the losers on the list of children of the winner depends on which link occurs first.)

We break the time during which links are done in a delete-min into two phases: the \emph{non-central phase}, during which the non-central links are done, and the \emph{central phase}, during which the central links are done.  We define the masses of roots differently during the two phases.

The links done during the non-central phase build all of the treap except the central path.  The part of the treap built during this phase consists of one subtree rooted at each vertex on the central path.  The nodes in such a subtree are either all left roots or all right roots.  If the root of a subtree is a left root, the subtree consists either of just this root or of a left link won by this root, plus the subtree of the treap rooted at the loser of this left link.  The situation is symmetric for subtrees rooted at right roots.   

During the non-central phase, we partition the subtrees of left (right, respectively) roots built by this phase into maximal paths of real right (left) links.  Each node on such a path is \emph{heavy} if its top node is a phantom child or is on the central path, \emph{light} otherwise.  A real left left-side (right right-side, respectively) link is \emph{anomalous} if its winner is heavy and the subtree of the treap rooted at its loser contains a k-link.  (Such links are non-central.)   See Figure~\ref{F:smooth-linking-example}.

\begin{figure}
\centering
\includegraphics[width=4.8in]{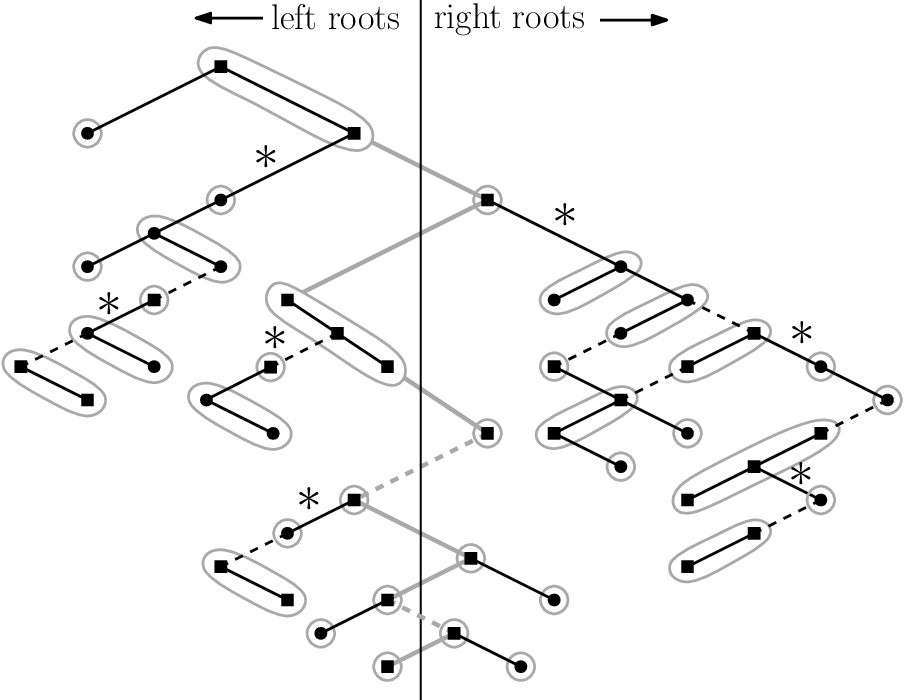}
% \Description{An example treap to illustrate the terms we have just defined for a smooth heap: heavy nodes, light nodes, and anomalous links.}
\caption{An example of the treap formed by a delete-min in a smooth heap.  Squares denote heavy nodes, circles light nodes.  Asterisks indicate anomalous links.  Solid links are real, dashed links phantom.  Crossing links are gray, left-side and right-side links black.  Each maximal path of real right left-side links or real left right-side links is circled.  During the non-central phase, all nodes on the central path are heavy.}
\label{F:smooth-linking-example}
\end{figure}

We apply the methods and results of Section~\ref{S:slim-analysis} separately to non-central left-side links and to non-central right-side links.  We define sizes and extended sizes exactly as for slim heaps: See Section~\ref{S:slim-ranks}.  We define masses of left children with respect to left children only, those of left roots with respect to left roots only, and symmetrically for right children and right roots.  Specifically, the mass of a temporary child is one plus the sum of the sizes of it and its interior siblings.

We define masses of roots similarly. The mass of a temporary root not in the middle of a delete-min is its size.  During the non-central linking phase of a delete-min, the mass of a light temporary root is its size plus the sum of the extended sizes of the roots interior to it on the root list.  The mass of a heavy temporary root is its extended size plus the sum of the extended sizes of the roots interior to it on the root list. There are two exceptions to this definition of root mass, which are the same as in slim heaps.
\begin{itemize}
\setlength\itemsep{1em}
\item[(i)] Just before a temporary node $w$ loses a real right (left, respectively) link $vw$ to a left (right) root $v$, its mass becomes $v.m'$, the mass of $v$ just after the link.  After the link, the mass of $w$ becomes the value defined above for a real child.
\item[(ii)] Just after a temporary node $v$ wins an anomalous link, its mass temporarily decreases to its size, then increases to the value defined above for a heavy root. 
\end{itemize}

As in slim heaps, we call the increases in mass and rank in exceptions (i) and (ii) \emph{step-ups}, and the decreases in mass and rank in exception (ii) \emph{step-downs}.  

Our main challenge in the analysis of smooth heaps is to bound the changes in node ranks caused by central links.  The problem is that a crossing link in effect moves the size of the loser from one side of the central boundary to the other, and we need to bound the effect of such back-and-forth movements.  This requires some care, and some careful definitions.

At the beginning of the central linking phase, the only roots remaining are those on the central path. During this phase, we define the masses of roots based on the \emph{real subpaths} of the central path.  These are the maximal subpaths of the central path formed by the real central links. Every permanent node on the central path constitutes its own real subpath, since it does not win or lose a real link. Given a real subpath $P$ with bottommost node $v$, let $M_L$ ($M_R$, respectively) be the sum of the sizes of all left (right) roots on the central path that are proper descendants in $\cT$ of $v$, where the sizes are evaluated at the start of the central linking phase.
We say $P$ is \emph{left-light} if $M_L \leq M_R$, otherwise it is \emph{right-light}.  
% We say $P$ is \emph{left-light} (\emph{right-light}, respectively) if $M_L \leq M_R$ ($M_L > M_R$).  
% C: Minor rephrasal
Let $M = \min\{M_L, M_R\}$ and $\barM = \max\{M_L, M_R\}$.  Let $P$ be a left-light path, and let $v$ be the topmost left node on $P$, if any.  We call $v$ and each node below $v$ on $P$ \emph{inner}.  We call each node above $v$ on $P$ \emph{outer}.  Each outer node on a left-light path is a right root.  The symmetric definitions apply to right-light paths.  

A path $P$ may or may not contain some outer nodes.  If it does, these nodes are all at the top of the path.  If $P$ is left-light, its outer nodes form a subpath of right links of right roots, and symmetrically if $P$ is right-light.  We define masses slightly differently for inner and outer nodes.

At the start of the central linking phase, we redefine the masses of the nodes on the central path, which are the remaining roots.  Let $v$ be a node on a maximal real subpath $P$ of the central path.  Let $M$ and $\barM$ be as defined above.  If $v$ is inner (outer, respectively), its mass becomes $M$ ($\barM$) plus its size plus the sum of the sizes of all roots on $P$ that are interior to $v$.  When $v$ wins a link during the central linking phase, its mass becomes $M$ ($\barM$) plus its new size.  (Just after $v$ wins such a link, no remaining root is interior to it.)    

Exception (i) applies to all real central links: Just before losing a real central link $vw$, the mass of $w$ steps up to $v.m'$, the mass of the winner after the link.  No central link is anomalous, so exception (ii) does not apply.

We define node ranks and link ranks for real deletion links exactly as in slim heaps.  In particular, if $vw$ is a real right left-side link, a real left right-side link, or a central link, $w.r_v$ is the rank of $w$ after its step-up just before the link, and if $vw$ is an anomalous link, $v.r_w'$ is the rank of $v$ after its step-down just after the link, but before its step-up.

Now we check the various results for slim heaps to determine how they apply to smooth heaps.  We begin by verifying that link ranks are non-negative.

The proof of Lemma~\ref{L:slim-link-ranks-nonnegative} holds for each real deletion link $vw$ whose winner link $uv$ is won by a left root $u$.  This includes the case in which $uv$ is a central link, since exception (i) still applies.  The symmetric argument holds for each real deletion link whose winner link is won by a right root.  We conclude that Lemma~\ref{L:slim-link-ranks-nonnegative} holds for smooth heaps: All link ranks are non-negative.

\subsection{Small-rank links in smooth heaps}\label{S:small-rank-smooth}

Next, we verify that the analysis for links of small rank in slim heaps applies to smooth heaps.  Indeed it does, if we adjust some constant factors and account for central links.

The proof of Lemma~\ref{L:slim-small-node-ranks} applies separately to the left roots and to the right roots just after the deletion of the original root in a delete-min.  Thus the number of roots of rank at most 1 is at most twice the bound in Lemma~\ref{L:slim-small-node-ranks}, specifically at most $5\lg n + 2$.
The proof of Lemma~\ref{L:multipass-small-link-rank-decreases} holds with one small change\footnote{The required change is that $2^{v.r'}$ is \emph{at most} $\log_2((u.s+v.s)/v.s)$, but not necessarily equal to it. This does not affect the validity of the proof.} for each real deletion link in a smooth heap whose winner link is a left-side non-central link, and holds symmetrically for each real deletion link whose winner link is a right-side non-central link.
% The proof of Lemma~\ref{L:multipass-small-link-rank-decreases} holds without change for each real deletion link in a smooth heap whose winner link is a left-side non-central link, and holds symmetrically for each real deletion link whose winner link is a right-side non-central link.
An inspection of the proof shows that it also holds for each real deletion link whose winner link is a central link, even with the redefinition of mass at the beginning of the central linking phase.  It follows that Lemma~\ref{L:slim-small-link-rank-decreases} holds for smooth heaps.  The proof of Lemma~\ref{L:slim-small} holds for slim heaps using twice the bound in Lemma~\ref{L:slim-small-node-ranks}, giving us a bound of at most six times the sum of the magnitudes of all node-rank decreases plus two per insertion plus two per decrease-key plus $30\lg n + 12$ per delete-min on the number of real deletion links of rank less than $1$ in any sequence of slim heap operations starting with no heaps.

Suppose Theorem~\ref{T:slim-node-rank-changes} holds for smooth heaps; that is, in any sequence of smooth heap operations starting with no heaps, the sum of the magnitudes of all node-rank changes is $\OO(\lg\lg n)$ per decrease-key plus $\OO(\lg n)$ per delete-min.  Then we can conclude that Theorem~\ref{T:slim-small} holds for smooth heaps: In any sequence of slim heap operations starting with no heaps, the total number of real deletion links with rank less than $1$ is $\OO(1)$ per insertion plus $\OO(\lg\lg n)$ per decrease-key plus $\OO(\lg n)$ per delete-min.

\subsection{Large-rank links in smooth heaps}\label{S:large-rank-smooth}

Continuing, we verify that the analysis for links of large rank in slim heaps applies to smooth heaps, assuming that Theorem~\ref{T:slim-node-rank-changes} holds for smooth heaps. 

The proof of Lemma~\ref{L:slim-v-shifts}, the $v$-shift bound for slim heaps, holds without change for smooth heaps.  Thus the sum of the $v$-shifts of all real deletion links is $\OO(\lg\lg n)$ per decrease-key plus $\OO(\lg n)$ per delete-min.

The $u$-shift bound for slim heaps in Lemma~\ref{L:slim-u-shifts} also holds for smooth heaps, but the proof requires some changes.  There are no root swaps after deletion of the original root, so the part of the proof concerning such swaps can be omitted.  On the other hand, the roots on the root list after deletion of the original root $u$ at the beginning of a delete-min are not in order by link time: The left roots are in decreasing order by link time and the right roots are in increasing order by link time.

Let $v_1, v_2,\dots, v_k$ be the temporary nodes in $\cT$, in symmetric order in $\cT$, and let $v_j$ be the child that $u$ acquired earliest.  Root $v_j$ is either the rightmost left root or the leftmost right root among the $v_i$.  For $1 \leq i \leq k$ let $u.r'_i$ be the value of $u.r$ just after $u$ lost the link $uv_i$ cut by the delete-min.  We redefine $\Delta_i$ to flip the sign of $\Delta_i$ for $i \geq j$.  Specifically, for $1 \leq i < j$,  $\Delta_i = u.r'_i - u.r'_{i+1}$; for $j \leq i < k$, let $\Delta_i = u.r'_{i+1} - u.r'_i$.  For $1 \leq i < k$,  let $I_i$ be the interval of time between just after $uv_i$ is done and just after $uv_{i+1}$ is done.  A given instant of time is in at most one $I_i$ for $i < j$ and one $I_i$ for $i \geq j$.  It follows that the sum of $|\Delta_i|$ for $1 \leq i <k$ is at most twice the sum of the changes in $u.r$ over all such changes.  The rest of the proof remains the same, except that the multiplying factor for intervals overlapping is two rather than three.  We conclude that the sum of the $u$-shifts of all real deletion links is $\OO(\lg\lg n)$ per decrease-key plus $\OO(\lg n)$ per delete-min.

The proofs of Theorem~\ref{T:slim-link-rank-sum},~\ref{T:slim-deletion-link-bound}, and~\ref{T:slim-time-bound} hold without change for smooth heaps if we replace each slim heap result with the corresponding result for smooth heaps.  That is, any sequence of smooth heap operations starting with no heaps takes takes $\OO(\lg\lg n)$ amortized time per decrease-key, $\OO(\lg n)$ amortized time per deletion, and $\OO(1)$ amortized (and worst-case) time for each other heap operation. 

\subsection{Node-rank changes in smooth heaps}

The last step is to prove Theorem~\ref{T:slim-node-rank-changes} for smooth heaps. The proofs of Lemmas~\ref{L:slim-real-decrease-key-links} and~\ref{L:slim-real-insertion-links} hold without change for smooth heaps.  Thus the sum of node rank increases caused by insertion and decrease-key links is at most $2\lg\lg n$ per decrease-key plus at most $\lg\lg n + 3$ per delete-min.

The proof of Lemma~\ref{L:slim-phantom-link-winners} holds for non-central phantom left-side links. The symmetric argument applies to non-central phantom right-side links.  We conclude that the sum of node rank increases caused by such links is at most $2\lg\lg n$ per decrease-key plus at most $6$ per delete-min.

The proof of Lemma~\ref{L:slim-real-link-winners} holds for real non-central left-side links, with one change: Such a link $vw$ decreases the mass of $w$ to \emph{at most} the size of $v$.  This does not affect the validity of the proof.  The symmetric argument applies to real non-central right-side links.  We conclude that real non-central left-side and right-side links do not increase any node rank, although the associated step-ups might.

The proof of Lemma~\ref{L:slim-anomalous-step-ups} holds for anomalous left-side links, and the symmetric argument applies to anomalous right-side links.  All anomalous links are non-central.  We conclude that the sum of the increases in node ranks caused by step-ups during anomalous links is at most $2\lg\lg n$ per decrease-key plus at most $6$ per delete-min.

The proof of Lemma~\ref{L:slim-heavy-step-ups} holds for real right non-central links between heavy left roots with trivial changes: There are no node swaps at the start of the delete-min, and the sum performed is just over the relevant left-side links. The symmetric argument applies to real left non-central links between heavy right roots.  We conclude that the sum of the increases in node ranks caused by step-ups before real non-central links between heavy roots is at most $4\lg n$.

The proofs of Lemmas~\ref{L:slim-light-right-links-some} and~\ref{L:slim-light-right-links-all} hold for real right links of light left roots, since the link $uv$ in the statement of Lemma~\ref{L:slim-light-right-links-some} is lost by a light node and hence cannot be a central link.  The symmetric argument applies to real left links of light right roots.  We conclude that the sum of rank increases caused by real non-central links between light nodes is at most $4\lg n$ per delete-min.

It remains to bound the sum of increases in node ranks caused by central links.

\begin{lemma}\label{L:crossing-link-node-rank-increases}
In any sequence of slim heap operations starting with no heaps, the sum of the increases in node ranks caused by central links is at most $\lg\lg n$ per decrease-key plus $10\lg n + \lg\lg n + 3$ per delete-min.
\end{lemma}
\begin{proof}
For a root $v$ on the central path, let unprimed variables $v.r$, $v.m$, $v.s$ take their values at the start of the central linking phase, and let primed variables $v.r'$, $v.m'$, $v.s'$ take their values just after $v$ wins its central link. If $v$ is the bottommost node on the central path, it does not win a central link and the primed and unprimed variables are equal.

At the start of the central linking phase, the mass of a root (but not is size) can change, but any such change is a decrease.  Hence such changes can only decrease root ranks.  A phantom central link does not change the rank of the winner and the loser's rank becomes zero. Hence only real central links can cause rank increases.

The winner of a real central link increases in mass, but its size increases by the same amount, so its rank does not increase. The loser's rank can increase due to the step-up. Nothing else causes rank increases during the central linking phase. It remains to bound the rank increases due to step-ups.

Let $vw$ be a real central link. Let $M$ and $\barM$ be with respect to the real subpath containing $v$ and $w$.

If $v$ and $w$ are outer, then the mass of $w$ steps up from $w.m'$ to $v.m' = w.m' + v.s$. The corresponding rank increase is at most $2\lg v.m' - 2\lg w.m'$. Observe that $v.m' - w.m' = v.m - w.m = v.s$, and that $v.m' \geq v.m$ and $w.m' \geq w.m$. Hence $2\lg v.m' - 2\lg w.m' \leq 2\lg v.m - 2\lg w.m$. Since $w$ is next to $v$ and interior to it at the start of the central linking phase, the sum of $2\lg v.m - 2\lg w.m $ over all central left-side links between outer roots is at most $2\lg n$. The same applies to the central right-side links between outer roots. Thus the sum of rank increases due to central links between outer roots is at most $4\lg n$.

If $v$ and $w$ are inner, the mass of $w$ steps up from $w.m'$ to $v.m' = w.m' + v.s$, yielding a rank increase of at most $2\lg v.m' - 2\lg w.m'$. Summing this bound over all links between inner roots within a single real subpath, the sum of rank increases is at most $2\lg x.m' - 2\lg M$, where $x$ is the highest inner root on the real subpath. We rewrite this quantity as $(2\lg x.m' - 2\lg x.m) + (2\lg x.m - 2\lg M)$, and bound the two terms separately.

We charge the first term, $2\lg x.m' - 2\lg x.m$, to the future deletion of $x$. Since $x.m' = M + x.s'$ and $x.m \geq M + x.s$, $2\lg x.m' - 2\lg x.m \leq 2\lg(M + x.s') - 2\lg(M + x.s)\leq 2\lg x.s' - 2\lg x.s$. The sum of $2\lg x.s' - 2\lg x.s$ over all changes in the size of $x$ is at most $2\lg n$. The total charge is $2\lg n$ per delete-min.

We charge the second term, $2\lg x.m - 2\lg M$, to the current delete-min. Suppose the real subpath containing $x$, say $P$, is left-light. If there is another left-light real subpath $P'$ above $P$ on the central path, then $x.m \leq M'$ where $M'$ is defined with respect to $P'$. This inequality also holds for the topmost left-light real subpath if we set $M' = n$.  Then $2\lg x.m - 2\lg M \leq 2\lg M' - 2\lg M$. The sum of this bound over all left-light real subpaths $P$ is at most $2\lg n$. The symmetric argument yields the same bound for the right-light real subpaths, giving a total bound of $4\lg n$. Hence the sum of rank increases due to central links between inner nodes is at most $6\lg n$ per delete-min.

Finally, on each real subpath there is at most one link between an outer root and an inner root. The step-up increases the rank of the inner root by at most $\lg\lg n$. We charge this increase to the phantom k-link lost by the top node on the real subpath, or to the current delete-min if the top node is the root of $\cT$. The total cost is $\lg\lg n$ per decrease-key plus $\lg\lg n + 3$ per delete-min, by Lemma~\ref{L:k-link-shift}. This proves the lemma.
\end{proof}

Combining our bounds, we obtain Theorem~\ref{T:slim-node-rank-changes} for smooth heaps, with somewhat larger constants.  That is, in any sequence of smooth heap operations starting with no heaps, the sum of the magnitudes of node-rank changes is $\OO(\lg\lg n)$ per decrease-key plus $\OO(\lg n)$ per delete-min.

\section{Lazy self-adjusting heaps}\label{S:lazy}

We have analyzed the eager (tree) versions of multipass pairing heaps, slim heaps, and smooth heaps.  The same bounds hold for the lazy (forest) versions of these data structures.  As described at the end of Section~\ref{S:canonical-framework}, these versions maintain a heap as a set of trees rather than just a single tree, represented by a circular list of roots and accessed by a root of minimum key, the \emph{min-root}.  Insertions, melds, and decrease-keys do not do links; instead, an insert makes the new node into a one-node tree and adds it to the root list, a meld catenates the root lists of the two heaps, and a decrease-key of a child node $w$ cuts $w$ from its parent and adds $w$ to the root list.  In each case, the operation updates the min-root.  A delete-min deletes the min-root and catenates the list of its children with the root list, making its children into \emph{new roots}. The other roots, which were already roots before the min-root was deleted, are \emph{old roots}. It then proceeds to link the roots using the appropriate linking method until one root remains.

In discussing the lazy versions, we shall assume that the list of old roots in a delete-min is added to the \emph{front} of the list of new roots, so that all the old roots precede all the new roots before any links occur.  Equivalently, we maintain the min-root as the rightmost root on the root list, and when a delete-min occurs the list of new roots replaces the deleted min-root.  In a smooth heap, we define each old root to be a left root.  

The only change we need to make in any of the analyses is to account for increases in node ranks of the old roots after deletion of the min-root.  Before the min-root is deleted, each old root has rank $0$.  Deletion of the min-root causes the mass of each old temporary root $v$ to increase from its size to the sum of its size and those of all temporary roots to its right on the root list.  This increases its mass by at most $\lg\lg n$.  Root $v$ became an old root by being inserted or by being cut from its parent in a decrease-key.  We charge the increase in the rank of $v$ to the next operation that makes $v$ a root not in the middle of a delete-min.  This operation is either a decrease-key or a delete-min.  An adaptation of the proof of Lemma~\ref{L:multipass-real-insertion-links} shows that the sum of the increases of the ranks of old roots at the beginning of delete-mins is $\OO(\lg n)$ per delete-min plus $\OO(\lg\lg n)$ per delete-min.  It follows that the bounds for the eager heap implementations hold for the lazy implementations as well.

Our bounds also hold if delete-min catenates the list of new roots to the front of the list of old roots rather than to the back but (except for smooth heaps) the analysis is much less straightforward.  For smooth heaps, we merely define old roots to be right roots instead of left roots.  For multipass pairing heaps and slim heaps, we define new roots to be left roots and old roots to be right roots, and analyze links between left and right roots in a way similar to what we did in the analysis of smooth heaps.

\section{Remarks}\label{S:remarks}

Using powerful new techniques, we have significantly improved the amortized time bounds for multipass pairing heaps, slim heaps, and smooth heaps.  Slim and smooth heaps are subject to the lower bounds of Iacono and {\"O}zkan, and our upper bounds match their lower bounds.  They are not subject to Fredman's lower bounds, but they do match them.  Multipass pairing heaps are subject to both sets of lower bounds.  Our bounds match them except for that of decrease-key, which differs from the lower bound by a factor of $\lg\lg\lg n$.  Our bounds hold for both the eager and the lazy versions of the data structures.

By showing that slim and smooth heaps match the lower bounds of Fredman and of Iacono and {\"O}zkan and are subject to Iacono and {\"O}zkan's bounds (though not Fredman's), we have closed a chapter in the design and analysis of self-adjusting heaps.  But this is far from resolving all interesting questions.  In particular, for decrease-key operations in multipass pairing heaps there is still a $\Theta(\lg\lg\lg n)$ gap between the upper and lower time bounds.  The known time-bound gap for decrease-key operations in classical pairing heaps is significantly larger, as we discussed in the introduction.  Whether our methods, or any others, will suffice to close these gaps is a subject for future study.  In particular, we are optimistic that our methods can be extended to give a much tighter bound for decrease-key in classical pairing heaps.

What is it that makes a self-adjusting heap efficient?  Our results suggest two requirements: The links in delete-min operations should respect link-time order, and each root should win at most a constant number of links in each delete-min.  We do not know whether these conditions are necessary: pairing heaps and their variants satisfy the first but not the second.  On the other hand, the assembly pass in pairing heaps (see~\cite{SimplerPairing, phd2022}), during which a root can win an arbitrary number of links, seems benign to us.  We think classical pairing heaps have the same amortized bounds as slim and smooth heaps, but we have no proof yet. 

Another possible research direction is to apply our methods to splay trees or similar binary search trees.  A previous paper on multipass pairing heaps~\cite{DKKPZ} connected their analysis to that of a variant of splay trees called \emph{path-balanced binary search trees}.  Perhaps our techniques could improve the analysis of this data structure as well. 

\bibliographystyle{plain}
\bibliography{main}

\providecommand{\noopsort}[1]{}
\begin{thebibliography}{10}

\bibitem{BrodalSurvey}
Gerth~St{\o}lting Brodal.
\newblock A survey on priority queues.
\newblock In Andrej Brodnik, Alejandro L{\'{o}}pez{-}Ortiz, Venkatesh Raman, and Alfredo Viola, editors, {\em Space-Efficient Data Structures, Streams, and Algorithms - Papers in Honor of J. Ian Munro on the Occasion of His 66th Birthday}, volume 8066 of {\em Lecture Notes in Computer Science}, pages 150--163. Springer, 2013.

\bibitem{DHIKP09}
Erik~D. Demaine, Dion Harmon, John Iacono, Daniel~M. Kane, and Mihai {P\v{a}tra\c{s}cu}.
\newblock The geometry of binary search trees.
\newblock In {\em {SODA} 2009}, pages 496--505, 2009.

\bibitem{DKKPZ}
Dani Dorfman, Haim Kaplan, L{\'{a}}szl{\'{o}} Kozma, Seth Pettie, and Uri Zwick.
\newblock Improved bounds for multipass pairing heaps and path-balanced binary search trees.
\newblock In Yossi Azar, Hannah Bast, and Grzegorz Herman, editors, {\em 26th Annual European Symposium on Algorithms, {ESA} 2018, August 20-22, 2018, Helsinki, Finland}, volume 112 of {\em LIPIcs}, pages 24:1--24:13. Schloss Dagstuhl - Leibniz-Zentrum f{\"{u}}r Informatik, 2018.

\bibitem{Elmasry09}
Amr Elmasry.
\newblock Pairing heaps with ${O}(\log\log n)$ decrease cost.
\newblock In {\em Proceedings of the twentieth Annual ACM-SIAM Symposium on Discrete Algorithms}, pages 471--476. SIAM, 2009.

\bibitem{Elmasry17}
Amr Elmasry.
\newblock Toward optimal self-adjusting heaps.
\newblock {\em {ACM} Trans. Algorithms}, 13(4):55:1--55:14, 2017.

\bibitem{Fox11}
Kyle Fox.
\newblock Upper bounds for maximally greedy binary search trees.
\newblock In {\em Workshop on Algorithms and Data Structures}, pages 411--422. Springer, 2011.

\bibitem{FredmanLB}
Michael~L. Fredman.
\newblock On the efficiency of pairing heaps and related data structures.
\newblock {\em J. {ACM}}, 46(4):473--501, 1999.

\bibitem{FSST86}
Michael~L. Fredman, Robert Sedgewick, Daniel~Dominic Sleator, and Robert~Endre Tarjan.
\newblock The pairing heap: {A} new form of self-adjusting heap.
\newblock {\em Algorithmica}, 1(1):111--129, 1986.

\bibitem{Fibonacci}
Michael~L. Fredman and Robert~Endre Tarjan.
\newblock Fibonacci heaps and their uses in improved network optimization algorithms.
\newblock {\em J. {ACM}}, 34(3):596--615, 1987.

\bibitem{HollowHeaps}
Thomas~Dueholm Hansen, Haim Kaplan, Robert~E. Tarjan, and Uri Zwick.
\newblock Hollow heaps.
\newblock {\em {ACM} Trans. Algorithms}, 13(3):42:1--42:27, 2017.

\bibitem{HKST21}
Maria Hartmann, L\'{a}szl\'{o} Kozma, Corwin Sinnamon, and Robert~E. Tarjan.
\newblock {Analysis of Smooth Heaps and Slim Heaps}.
\newblock In Nikhil Bansal, Emanuela Merelli, and James Worrell, editors, {\em 48th International Colloquium on Automata, Languages, and Programming (ICALP 2021)}, volume 198 of {\em Leibniz International Proceedings in Informatics (LIPIcs)}, pages 79:1--79:20, Dagstuhl, Germany, 2021. Schloss Dagstuhl -- Leibniz-Zentrum f{\"u}r Informatik.

\bibitem{IaconoPairing}
John Iacono.
\newblock Improved upper bounds for pairing heaps.
\newblock In Magn{\'{u}}s~M. Halld{\'{o}}rsson, editor, {\em Algorithm Theory - {SWAT} 2000, 7th Scandinavian Workshop on Algorithm Theory, Bergen, Norway, July 5-7, 2000, Proceedings}, volume 1851 of {\em Lecture Notes in Computer Science}, pages 32--45. Springer, 2000.

\bibitem{IaconoOzkan}
John Iacono and {\"O}zg{\"u}r {\"O}zkan.
\newblock A tight lower bound for decrease-key in the pure heap model.
\newblock {\em arXiv preprint arXiv:1407.6665}, 2014.

\bibitem{KaplanTZ14}
Haim Kaplan, Robert~Endre Tarjan, and Uri Zwick.
\newblock Fibonacci heaps revisited.
\newblock {\em CoRR}, abs/1407.5750, 2014.

\bibitem{KS19}
L\'{a}szl\'{o} Kozma and Thatchaphol Saranurak.
\newblock Smooth heaps and a dual view of self-adjusting data structures.
\newblock {\em SIAM Journal on Computing}, 49(5):STOC18--45--STOC18--93, 2020.

\bibitem{LarkinSenTarjan}
Daniel~H. Larkin, Siddhartha Sen, and Robert~Endre Tarjan.
\newblock A back-to-basics empirical study of priority queues.
\newblock In Catherine~C. McGeoch and Ulrich Meyer, editors, {\em 2014 Proceedings of the Sixteenth Workshop on Algorithm Engineering and Experiments, {ALENEX} 2014, Portland, Oregon, USA, January 5, 2014}, pages 61--72. {SIAM}, 2014.

\bibitem{Luc88}
Joan~M. Lucas.
\newblock Canonical forms for competitive binary search tree algorithms.
\newblock {\em Tech. Rep. DCS-TR-250, Rutgers University}, 1988.

\bibitem{mehlhorn2015greedy}
Kurt Mehlhorn and Thatchaphol Saranurak.
\newblock Greedy is an almost optimal deque.
\newblock In {\em Algorithms and Data Structures: 14th International Symposium, WADS 2015, Victoria, BC, Canada, August 5-7, 2015. Proceedings}, volume 9214, page 152. Springer, 2015.

\bibitem{Mun00}
J.Ian Munro.
\newblock On the competitiveness of linear search.
\newblock In Mike~S. Paterson, editor, {\em Algorithms - ESA 2000}, volume 1879 of {\em Lecture Notes in Computer Science}, pages 338--345. Springer Berlin Heidelberg, 2000.

\bibitem{PettiePairing}
Seth Pettie.
\newblock Towards a final analysis of pairing heaps.
\newblock In {\em 46th Annual {IEEE} Symposium on Foundations of Computer Science {(FOCS} 2005), 23-25 October 2005, Pittsburgh, PA, USA, Proceedings}, pages 174--183. {IEEE} Computer Society, 2005.

\bibitem{phd2022}
Corwin Sinnamon.
\newblock {\em Analysis of self-adjusting heaps}.
\newblock PhD thesis, Princeton University, 2022.

\bibitem{SimplerPairing}
Corwin Sinnamon and Robert Tarjan.
\newblock A simpler proof that pairing heaps take {$O(1)$} amortized time per insertion, 2022.

\bibitem{SODAMultipass}
Corwin Sinnamon and Robert~E Tarjan.
\newblock A nearly-tight analysis of multipass pairing heaps.
\newblock In {\em Proceedings of the 2023 Annual ACM-SIAM Symposium on Discrete Algorithms (SODA)}, pages 535--548. SIAM, 2023.

\bibitem{SODASmoothSlim}
Corwin Sinnamon and Robert~E Tarjan.
\newblock A tight analysis of slim heaps and smooth heaps.
\newblock In {\em Proceedings of the 2023 Annual ACM-SIAM Symposium on Discrete Algorithms (SODA)}, pages 549--567. SIAM, 2023.

\bibitem{tarjan1983data}
Robert~Endre Tarjan.
\newblock {\em Data structures and network algorithms}.
\newblock SIAM, 1983.

\end{thebibliography}
\end{document}